\documentclass[12pt]{article}
\usepackage{}
\usepackage{mathrsfs}
\usepackage{mathrsfs}
\usepackage{}
\usepackage{amsfonts}
\usepackage{amsmath,amssymb,bm,graphicx,latexsym,graphics,epsfig,subfigure,color}
\usepackage{amsthm}
\usepackage{bm}   
\usepackage{amsmath,amsfonts,amssymb} 
\usepackage{bbm}
\usepackage{color}
\usepackage{graphicx}
\usepackage{epstopdf}
\usepackage{epsfig}
\usepackage{setspace}
\usepackage{indentfirst}
\usepackage{booktabs}
\usepackage{float}
\usepackage[utf8]{inputenc}
\usepackage[title]{appendix}
\usepackage{amsthm}
\usepackage{mathrsfs}
\usepackage{multirow}
\usepackage{diagbox}
\usepackage[colorlinks=true,linkcolor=blue,urlcolor=blue,citecolor=blue]{hyperref} 
\usepackage{amssymb}  
\usepackage{geometry}
\usepackage[authoryear]{natbib}
\usepackage{hyperref}
\hypersetup{bookmarksnumbered, bookmarksopen,
colorlinks,citecolor=blue,linkcolor=blue,CJKbookmarks}

\newtheorem*{Lem*}{Lemma}
\newtheorem{The}{Theorem}

\newtheorem{Lem}{Lemma}[section]

 \numberwithin{The}{section}

\newcommand{\ignore}[1]{}

\newcommand{\beginsupplement}{%
        \setcounter{table}{0}
        \renewcommand{\thetable}{S\arabic{table}}%
        \setcounter{figure}{0}
        \renewcommand{\thefigure}{S\arabic{figure}}%
     }

\newtheorem{assump}{Assumption}

\theoremstyle{definition}

\newtheorem{remark}{Remark}
\newtheorem{algorithm}{Algorithm}

\usepackage{algorithm}
\usepackage{algorithmic} 

\def\hDash{\bot\!\!\!\bot}

\def\X{{\bf X}}

\newcommand{\E}{{\mathbb E}}

\def\cp{\overset{{p}}{\longrightarrow}}

\graphicspath{{./pictures/}}

\usepackage{xr}    
\usepackage{ hyperref}
\begin{document}

\date{}
\title
{\bf Model-free controlled variable selection via data splitting}
\maketitle
\vspace{-2.5cm}
\centerline{Yixin Han\textsuperscript{1}, Xu Guo\textsuperscript{2} \& Changliang Zou\textsuperscript{1}} \vspace{.4cm}     

\centerline{\it \textsuperscript{1}School of Statistics and Data Science, LPMC $\&$ KLMDASR, Nankai University}
         

\centerline{\it\textsuperscript{2}Department of Mathematical Statistics, Beijing Normal University}


\begin{abstract}

Addressing the simultaneous identification of contributory variables while controlling the false discovery rate (FDR) in high-dimensional data is a crucial statistical challenge. In this paper, we propose a novel model-free variable selection procedure in sufficient dimension reduction framework via a data splitting technique.  The variable selection problem is first converted to a least squares procedure with several response transformations. We construct a series of statistics with global symmetry property and leverage the symmetry to derive a data-driven threshold aimed at error rate control.  Our approach demonstrates the capability for achieving finite-sample and asymptotic FDR control under mild theoretical conditions. Numerical experiments confirm that our procedure has satisfactory FDR control and higher power compared with existing methods.

\end{abstract}

\section{Introduction}\label{Sec:Introduction}

Sufficient dimension reduction (SDR) is a powerful technique to extract relevant information from high-dimensional data \citep{li1991sliced,cook1991sliced,xia2002adaptive,li2007directional}. We use $Y$ with support $\Omega_Y$ to denote the univariate response, and let $\X=(X_1,\ldots,X_p)^{\top}\in\mathbb{R}^p$ be the $p$-dimensional vector of all covariates. The basic idea of SDR is to replace the predictor vector with its projection onto a subspace of the predictor space without loss of information on the conditional distribution of $Y$ given $\X$. 
In practice, a large number of features in high-dimensional data are typically collected, but only a small portion of them are truly associated with the response variable. However, while grasping important features or patterns in the data, the reduction subspace from SDR usually includes all original variables which makes it difficult to interpret. Therefore, in this paper, we aim at developing a model-free variable selection procedure to screen out truly non-contributing variables with certain error rate control, thus making the subsequent model building feasible or simplified and helping reduce the computational cost caused by high-dimensional data.

Let $F( Y\mid \X)$ denote the conditional distribution function of $Y$ given $\X$. 
The index sets of the active and inactive variables are defined respectively as
\begin{align*}
	\mathcal{A} &= \left\{j: F(Y\mid \X) ~\mbox{functionally depends on}~ X_j, j=1,\ldots, p\right\},\\
	\mathcal{A}^{c} &= \left\{j:  F(Y\mid \X)~\mbox{ does not functionally depend on}~ X_j,j=1,\ldots,p\right\}.
\end{align*}
Many prevalent variable selection procedures have been developed under the paradigm of linear models or generalized linear models, such as 
LASSO \citep{tibshirani1996regression}, SCAD \citep{fan2001variable}, or adaptive LASSO \citep{zou2006adaptive}. 
See the review of \citet{fan2010selective} and the book of  \citet{fan2020statistical} for a fuller list of references. In contrast, model-free variable selection can be achieved by SDR since it does not require complete knowledge of the underlying model, thus researchers 
can avoid disposing of model misspecification.

SDR methods with variable selection aim to find the active set $\mathcal{A}$ such that
\begin{align}\label{SVS}
	Y\hDash\X_{\mathcal{A}^c}\mid \X_{\mathcal{A}},
\end{align}
where ``$\hDash$'' stands for independence, $\X_{\mathcal{A}} = \left\{\X_{j}: j\in\mathcal{A}\right\}$ denotes the vector containing all active variables and $\X_{\mathcal{A}}^c$ is the complementary set of $\X_{\mathcal{A}}$. Condition \eqref{SVS} implies that $\X_{\mathcal{A}}$ contains all the relevant information in terms of predicting $Y$. \citet{li2005model} proposed to combine sufficient dimension reduction and variable selection. 
 \citet{chen2010coordinate} proposed a coordinate-independent sparse estimation that can simultaneously achieve sparse SDR and screen out irrelevant variables efficiently.  \citet{wu2011asymptotic} focused on the model-free variable selection with a diverging number of predictors. 
A marginal coordinate hypothesis is proposed by \citet{cook2004testing} for model-free variable selection under low-dimensional settings, and then is promoted by \citet{shao2007marginal} and \citet{yu2016model}.  \citet{yu2016marginal} constructed marginal coordinate tests for sliced inverse regression (SIR) and \citet{yu2016trace} suggested a trace-pursuit-based utility for ultrahigh-dimensional feature selection.  See  \citet{li2020selective} and \citet{zhu2020review} for a comprehensive review.

However, those existing approaches do not account for uncertainty quantification of the variable selection, i.e., the global error rate control in the selected subset of important covariates in high-dimensional situations.
In general high-dimensional nonlinear settings, \citet{candes2018panning} developed a Model-X Knockoff framework for controlling false discovery rate \citep[FDR,][]{benjamini1995controlling}, which was motivated by the pioneering Knockoff filter \citep{barber2015controlling}. 
Their statistics constructed via ``Knockoff copies" would satisfy (or roughly) joint exchangeability and thus can yield finite-sample FDR control.  However, the Model-X Knockoff requires knowing the joint distribution of the covariates, which is typically difficult in high-dimensional settings. 
Recently, \citet{guo2024model} improved the line of marginal tests  \citep{cook2004testing,yu2016model} by using decorrelated score type statistics to make inferences for a specific predictor which is of interest in advance. They further  leveraged the standard  \citet{benjamini1995controlling}  on $p$-values to control FDR, but the intensive computation of the decorrelated process may limit its application to high-dimensional situations. In a different direction, \cite{du2021false} 
proposed a data splitting strategy, named symmetrized data aggregation (SDA), to construct a series of statistics with global symmetry property and then utilize the symmetry to derive a data-driven threshold for error rate control. Specifically, \citet{du2021false} aggregated the dependence structure into a linear model with a pseudo response and a fixed covariate, making  the dependence structure become a blessing for power improvement. Similar to the Knockoff method, the SDA is also free of $p$-values and its construction does not rely on contingent assumptions, which motivates us to employ it in sufficient dimension reduction problems.

In this paper, we propose a model-free variable selection procedure that could achieve an effective FDR control. We first recast the problem of conducting variable selection in sufficient dimension reduction into making inferences on regression coefficients in a set of linear regressions with several response transformations. A variable selection procedure is subsequently developed via error rate control for low-dimensional and high-dimensional settings, respectively. 
Our main contributions include: 
(1) This novel data-driven selection procedure can control the FDR while being combined with different existing SDR methods for model-free variable selection by choosing different response transformation functions. (2) Our method does not need to estimate any nuisance parameters such as the structural dimension in SDR. (3) Notably, the proposed procedure is computationally efficient and easy to implement since it only involves a one-time split of the data and the calculation of the product of two dimension reduction matrices obtained from two splits. (4) Furthermore, this method can achieve finite-sample and asymptotic FDR control under some mild conditions. (5) Numerical experiments indicate that our procedure exhibits satisfactory FDR control and higher power compared with existing methods.

The rest of this paper is organized as follows. In section \ref{Methodology}, we present the problem and model formulation. In section \ref{variableselection}, we propose a low-dimensional variable selection procedure with error rate control and then discuss its extension in high-dimensional situations. The finite-sample and asymptotic theories for controlling the FDR are developed in Section \ref{theory}. Simulation studies and a real-data investigation are conducted in Section \ref{numerical} to demonstrate the superior performance of the proposed method. Section \ref{discussion} concludes the paper with several further topics. The main theoretical proofs are given in Appendix. More detailed proofs and additional numerical results are delineated in the Supplementary Material.

Notations. Let $\lambda_{\min}({\bf B})$ and $\lambda_{\max}({\bf B})$ denote the smallest and largest eigenvalues of square matrix ${\bf B}=(b_{ij})$. Write $\|{\bf B}\|_2=(\sum\nolimits_i\sum\nolimits_j b_{ij}^2)^{1/2}$ and
$\|{\bf{B}}\|_{\infty}=\max_i\sum\nolimits_{j}|b_{ij}|$. Denote $\|\bm\mu\|_1=\sum\nolimits_i|\mu_i|$ and $\|\bm\mu\|_2=(\sum\nolimits_i\mu_i^2)^{1/2}$ be the $L_1$ and $L_2$ norm of vector $\bm\mu$. Denote $\mathbb{E}(\X)$ and $\text{cov}(\X)$ be the expectation and covariance for random vector $\X$, respectively. 
Let $A_{n}\approx B_{n}$ denote that two quantities $A_{n}$ and $B_{n}$ are asymptotically equivalent, in the sense that there is a
constant $C>1$ such that $B_{n}/C\leq A_{n}\leq B_{n}C$ with probability tending to 1. The ``$\gtrsim$" and ``$\lesssim$" are similarly defined.

\section{Problem and model formulation}\label{Methodology}

The variable selection in~\eqref{SVS} can be framed as a  multiple testing problem
\begin{align}\label{orignaltest}
	\mathbb{H}_{0j}': j\in\mathcal{A}^{c} ~~\mbox{versus}~~
	\mathbb{H}_{1j}': j \in \mathcal{A}.
\end{align}
This is known as the marginal coordinate hypothesis described in \citet{cook2004testing} and \citet{yu2016marginal}. Some related works include \citet{li2005model}, \citet{shao2007marginal}, and \citet{yu2016model}. This type of selection procedure usually uses some nonnegative marginal utility statistics $W_j$'s to measure the importance of $X_j$'s to $Y$ in certain sense. However, the global error rate control within those methods is still challenging because the determination of selection thresholds generally involves the approximation to the distribution of $W_j$, and the accuracy of asymptotic distributions heavily affects the error rate control.

As a remedy, we consider a reformulation for \eqref{orignaltest}. Let $\bm\Sigma=\mathrm{cov}(\X)>0$  and assume $\mathbb{E}(\X)=0$. Denote $\mathcal{C}$ is dimension reduction subspace. For any function $f(Y)$ satisfying $\mathbb{E}\left\{f(Y)\right\}=0$, it has been demonstrated by \citet{yin2002dimension} and \citet{wu2011asymptotic} that
\begin{align*}
	\bm\Sigma^{-1}\text{cov}\left(\X,f(Y)\right)\in\mathcal{C},
\end{align*}
under the linearity condition \cite{li1991sliced}, which is usually satisfied when $\X$ is elliptical distribution. 
The transformation $f(\cdot)$ is used in a way different from its traditional role of being a mechanism for improving the goodness of model fitting. It serves as an intermediate tool for performing dimension reduction. Consequently, different transformation functions correspond to different SDR methods \citep{dong2021brief}. One can choose a series of transformation functions, $f_1(Y),\ldots, f_H(Y)$, whose forms do not depend on data. The $H\left(>d\right)$, a pre-specified integer, is usually called a working dimension and $d$ is the true structural dimension of the subspace $\mathcal{C}$. Given the working dimension $H$, at the population level, define
\begin{align}\label{OLS}
	\bm\beta_h^0 = \arg\min_{\bm\beta_h}\mathbb{E}\left[\left\{f_h(Y)-\X^{\top}\bm\beta_h\right\}^2\right], ~~h=1,\ldots,H.
\end{align}
Write $\textbf{B}_0 = \left(\bm\beta_1^0,\ldots,\bm\beta_H^0\right)$, then $\text{Span}(\textbf{B}_0)\subseteq \mathcal{C}$, and $\text{Span}(\textbf{B}_0)$ represents the subspace spanned by the column vector of $\textbf{B}_0$. By the following usual protocol in the literature of sufficient dimension reduction, we take one step further by assuming the coverage condition $\text{Span}(\textbf{B}_0)=\mathcal{C}$ whenever $\text{Span}(\textbf{B}_0)\subseteq\mathcal{C}$. This condition often holds in practice; see \citet{cook2006using} for further discussion.

For $j =1,\ldots, p$, let $\beta_{hj}^0$ be the $j$th element of $\bm\beta_h^0\in\mathbb{R}^{p}$, $h=1,\ldots, H$. If the $j$th variable is unimportant, $\textbf{B}_j=\textbf{0}\in\mathbb{R}^H$, where $\textbf{B}_j=(\beta_{1j}^0, \beta_{2j}^0, \cdots, \beta_{Hj}^0)^\top$
denotes the $j$th row of $\textbf{B}_0$. Further,
$Y\hDash \X \mid\X_{\mathcal{A}}$ implies that $\sum\nolimits_{h=1}^H|\beta_{hj}^0|>0$ for $j \in \mathcal{A}$ and $\sum\nolimits_{h=1}^H|\beta_{hj}^0|=0$ for $j\in\mathcal{A}^{c}$. In other words, if $j$ belongs to the active set $\mathcal{A}$, response $Y$ must depend on $X_j$ through at least one of the $H$ linear combinations. If $j$ belongs to the inactive set $\mathcal{A}^c$, none of the $H$ linear combinations involve $X_j$ \citep{yu2016marginal}. 	Accordingly, the testing problem \eqref{orignaltest} is equivalent to
\begin{align}\label{test}
	\mathbb{H}_{0j}: \sum\limits_{h=1}^H\left|\beta_{hj}^0\right|=0 ~~\mbox{versus}~~ \mathbb{H}_{1j}: \sum\limits_{h=1}^H\left|\beta_{hj}^0\right|>0.
\end{align}

Based on the above discussion, selecting active variables in a model-free framework is equivalent to selecting important variables in \emph{multiple response linear model}. Assume that there are independent
and identical distributed data $\mathcal{D}=\left\{\X_i,Y_i\right\}_{i=1}^{2n}$. 
Denote $\widehat{\bm\beta}_1,\ldots,\widehat{\bm\beta}_H$ are the estimators of $\bm\beta_1^{0},\ldots,\bm\beta_H^{0}$, and $W_j$'s as the marginal statistics based on the sample $\mathcal{D}$ associated with the variants of $\widehat{\bm\beta}_1,\ldots,\widehat{\bm\beta}_H$. Its explicit form would be given in the next section.
A selection procedure with a threshold $L$ is formed as
\begin{align}\label{selection}
	\widehat{\mathcal{A}}(L) = \left\{j: {W}_j\geq L, ~\mbox{for}~ 1\leq j\leq p\right\},
\end{align}
where $\widehat{\mathcal{A}}(L)$ is the estimate of $\mathcal{A}$ with threshold $L$. Obviously, $L$ plays an important role in variable selection to control the model complexity. We will construct an appropriate threshold by controlling the FDR to achieve model-free variable selection in SDR.

Denote $p_0 = \left|\mathcal{A}^c\right|$, $p_1 = \left|\mathcal{A}\right|$ and assume that $p_1$ is dominated by $p$, i.e., $p_1=o(p)$. The false discovery proportion (FDP) associated with the selection procedure \eqref{selection} is
\begin{align*}
	\mbox{FDP}\left(\widehat{\mathcal{A}}(L)\right)=\frac{\#\{j:j\in\widehat{\mathcal{A}}(L)\bigcap \mathcal{A}^c\}}{\#\{j:j\in\widehat{\mathcal{A}}(L)\}\vee 1},
\end{align*}
where $a\vee b=\max\left\{a,b\right\}$ and  $\#\{\}$ stands for the cardinality of an event. The FDR is defined as the expectation of the FDP, i.e., $\text{FDR}(L)=\mathbb{E}\left(\text{FDP}(L)\right)$. Our main goal is to find a data-driven threshold $L$ that controls the asymptotic FDR at a target level $\alpha$, 
	\begin{align*}
		\limsup_{n\to\infty} \mbox{FDR}\left(\widehat{\mathcal{A}}(L)\right)\leq \alpha.
	\end{align*}

\section{Variable selection via FDR control}\label{variableselection}

In this section, we first provide a data-driven variable selection procedure in Subsection \ref{lowdim} to control the FDR via data splitting technique in a model-free context when $p<n$, and the high-dimensional version is deferred in Subsection \ref{highdim}.

\subsection{Low-dimensional procedure}\label{lowdim}

We first split the full data $\mathcal{D}=\left\{\X_i, Y_i\right\}_{i=1}^{2n}$ into two independent parts $\mathcal{D}_1=\left\{\X_{1i}, Y_{1i}\right\}_{i=1}^{n}$ and $\mathcal{D}_2=\left\{\X_{2i}, Y_{2i}\right\}_{i=1}^{n}$ with equal size, which is respectively used to estimate the dimension reduction spaces as $\widehat{\textbf{B}}_1$ and $\widehat{\textbf{B}}_2$, where $\widehat{\textbf{B}}_1=(\widehat{\bm\beta}_1^{(1)},\ldots,\widehat{\bm\beta}_H^{(1)})$  and $\widehat{\textbf{B}}_2=(\widehat{\bm\beta}_1^{(2)},\ldots,\widehat{\bm\beta}_H^{(2)})$. 
One can find the unequal size data splitting investigation in \citet{du2021false}. On split $\mathcal{D}_k$, $k=1,2$, the least square estimator of ${\bf{B}}_j$ 
\begin{align*}
	\widehat{{\bf B}}_{kj}^\top = \bm{e}_j^{\top}\left(\sum\nolimits_{i=1}^n\X_{ki}\X_{ki}^{\top}\right)^{-1}\sum\limits_{i=1}^n\X_{ki}{\bf{f}}_{ki}^\top,
\end{align*}
where ${\bf f}=\left(f_1(Y),\ldots,f_H(Y)\right)^\top$ and $\bm{e}_j$ is the $p$-dimensional unit vector with the $j$th element being 1. The information from two parts is then combined 
to form a symmetrized ranking statistic
\begin{align}\label{rankstatistic}
	W_j = \frac{ \widehat{\textbf{B}}_{1j}^\top \widehat{\textbf{B}}_{2j}}{s_{1j}s_{2j}}, ~~ j=1,\ldots,p,
\end{align}
where $s_{kj}^2 = \bm{e}_j^{\top}\left(\sum\nolimits_{i=1}^n\X_{ki}\X_{ki}^{\top}\right)^{-1}\bm{e}_j$, $k=1,2$.
For an active variable, if $\sum\nolimits_{h=1}^H|\beta_{hj}^0|$ is large (under $\mathbb{H}_{1j}$), then both $\widehat{\textbf{B}}_{1j}$ and $\widehat{\textbf{B}}_{2j}$ have the same sign and tend to have large absolute values, thereby leading to a positive and large $W_j$ \citep{du2021false}. For a null feature, $W_j$ is symmetrically distributed around zero. It implies that $W_j$ demonstrates the marginal symmetry property \citep{barber2015controlling,du2021false} for all inactive variables such that it can be used to determine active or inactive variables. This motivates us to choose a data-driven threshold $L$ as the following to control the FDR at level $\alpha$
\begin{align}\label{threshold}
	L = \inf\left\{t >0: \frac{\#\left\{j: W_j\leq -t\right\}}{\#\left\{j:W_j\geq t\right\}\vee 1}\leq \alpha\right\},
\end{align}
If the above set is empty, we simply set $L = +\infty$. Then our decision rule is given by $\widehat{\mathcal{A}}(L)=\left\{j: W_j \geq L,1\leq j \leq p\right\}$. The fraction in \eqref{threshold} is an estimate of the FDP since $\#\left\{j: W_j \leq -t\right\}$ is a good approximation to $\#\left\{j: W_j\geq t, j\in \mathcal{A}^c\right\}$ by the marginal symmetry of $W_j$ under null. The core of our procedure is to construct marginal symmetric statistics using the data splitting technique, to obtain a data-driven threshold to realize variable selection.  Therefore, we refer our method  to {\bf \emph{Model-Free Selection via Data Splitting} (MFSDS)}. 

Since the estimators, $\widehat{\textbf{B}}_1$ and $\widehat{\textbf{B}}_2$, only are two approximations of $\textbf{B}$ and they are not derived by eigenvalue decomposition system \citep{li1991sliced,cook1991sliced}, there is no concern about that $\widehat{\textbf{B}}_1$ and  $\widehat{\textbf{B}}_2$ may not be in the same subspace. Our ultimate goal is to identify the active variables rather than to recover the dimension reduction subspace. It implies that variable selection achieved through \eqref{OLS} requires no dimension reduction basis estimation and thus is dispensable for the knowledge of the structural dimension $d$ either. 
Therefore, the proposed method can be adapted to a family of inverse slice regression estimators by choosing different $f(Y)$. In a nutshell, our method can be widely used due to its simplicity, computational efficiency, and generality. It is summarized as follows.

\begin{algorithm}[htbp]
	\caption{Model-free selection via data splitting (MFSDS)}
	\begin{algorithmic}\label{algorithmSDS}
		\STATE \textbf{Step 1} (Initialization) Specify $H$,  ${\bf f}$ and $\alpha$; \smallskip
		
		\STATE \textbf{Step 2} (Data splitting) Randomly split the data into two independent parts $\mathcal{D}_1$ and $\mathcal{D}_2$ with equal size. Obtain the dimension reduction estimates $\widehat{\textbf{B}}_1$ and $\widehat{\textbf{B}}_2$ by \eqref{OLS}; \smallskip
		
		\STATE \textbf{Step 3} (Ranking statistics)
		Construct the test statistics $W_j$ by \eqref{rankstatistic} and then rank them; \smallskip
		
		\STATE \textbf{Step 4} (Thresholding) Compute the threshold
		$L$ in \eqref{threshold}, and obtain selected variable set $\widehat{\mathcal{A}}(L)$. 
	\end{algorithmic}
\end{algorithm}

The total computational complexity of Algorithm \ref{algorithmSDS} is of order $O(2nHp^2 + p\log p)$ so that this algorithm can be easily implemented. Practically, our method involves data splitting that may lead to some information loss concerning the full data \citep{du2021false}.  Fortunately, we obtain a data-driven threshold by the marginal symmetry property of $W_j$ under the null, which does not need to find the null asymptotic distribution anymore. Here we use a toy example to illustrate the advantage of data splitting. Further details regarding the data generation can be found in Section \ref{numerical}. In Figure \ref{visualization}, we observe that the data splitting method (left panel) places most active variables above zero, and many inactive variables are symmetrically distributed around zero. This is a crucial property for our selection procedure while the full estimation (middle panel) and half data estimation (right panel) methods both fail to achieve this level of symmetry.

\begin{figure}[htbp]
	\centering		\includegraphics[height=5cm,width=15cm]{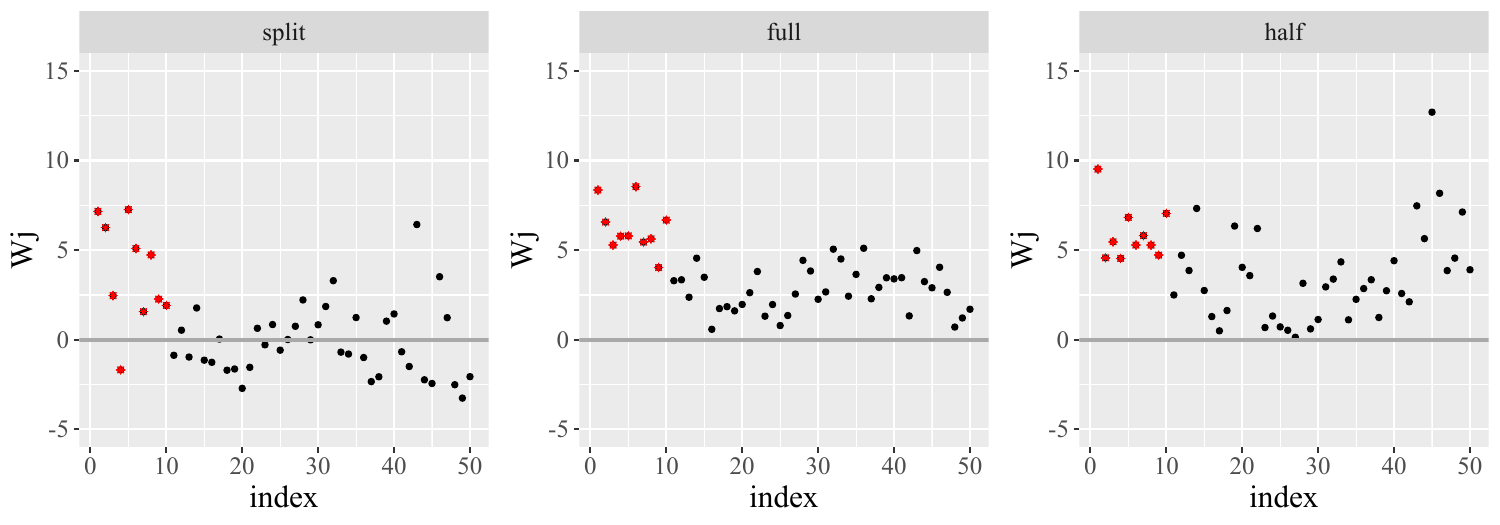}
	\caption{Scatterplots of $W_j$'s with red points and black dots denoting active and inactive variables, respectively. Left panel: the proposed $W_j$ in \eqref{rankstatistic}; Middle panel: $W_j = \sum\nolimits_{h=1}^H|\widehat\beta_{hj}|$ with $(\widehat{\textbf{B}})_{hj}=\widehat\beta_{hj}$, which is the least square estimator on the full data; Right panel: replace $\widehat{\textbf{B}}_2$ with $\widehat{\textbf{B}}_1$ in \eqref{rankstatistic}.
	}
	\label{visualization}
\end{figure}

\subsection{High-dimensional procedure }\label{highdim}

When the dimension $p$ is very large in practice, the above procedure does not work since the ordinary least square procedure cannot be directly implemented. Note that our data splitting procedure can be essentially extended to the version of the regularization form. Inspired by the idea of SDA filter proposed by \citet{du2021false}, we then develop  the following selection procedure for high-dimensional data.

To extract information from $\mathcal{D}_1$, we replace the least square solution in \eqref{OLS} with LASSO selector \citep{tibshirani1996regression} as follows
\begin{align}\label{sparseOLS}
	\widehat{\bm\beta}^{(1)}_{\lambda_h} = \arg\min_{\bm\beta_h}\left[n^{-1}\sum\limits_{i=1}^n\left\{f_h(Y_i)-\X_i^{\top}\bm\beta_h\right\}^2 + \lambda_h \|\bm\beta_h\|_1\right], ~~ h=1,\ldots,H,
\end{align}
where 
$\lambda_h>0$ is a tuning parameter. It is worth noting that although the LASSO estimator does not provide guarantees on the FDR control of the selected variables, it still serves as a useful tool here that simultaneously takes into account the sparsity and dependence structures as described in \cite{du2021false}. It is not necessary to penalize $H$ slices of coefficients simultaneously to establish ${\bf B}$. This is quite different from the traditional model-free variable selection methods, such as \citet{wu2011asymptotic}.
   
Let $\mathcal{S}=\{j:\sum\nolimits_{h=1}^H|\widehat{\beta}^{(1)}_{\lambda_h,j}|>0\}$ be the subset of variables selected by \eqref{sparseOLS}, where $\widehat{\beta}^{(1)}_{\lambda_h,j}$ is the $j$th element of $\widehat{\bm\beta}_{\lambda_h}$. 
We then use $\mathcal{D}_2$ to obtain the least square estimates $\widehat{\textbf{B}}_{2\mathcal{S}}$ in \eqref{OLS} for coordinates in the narrowed subset $\mathcal{S}$. Denote the estimates from $\mathcal{D}_1$ and $\mathcal{D}_2$ be $\widehat{\bf{B}}_{1}$ and $\widehat{\bf{B}}_{2}$, 
where
\begin{align*}
	{\widehat{\textbf{B}}}_{2j} =
	\begin{cases}
		{\widehat{\textbf{B}}}_{2\mathcal{S}j}, &\quad j\in\mathcal{S};\\
		\bm 0,& \quad \mbox{otherwise}.
	\end{cases}
\end{align*}
Accordingly, the ranking statistics in the high-dimensional setting are constructed as
\begin{align*}
	{W}_j = \frac{\widehat{\textbf{B}}_{1j}^\top \widehat{\textbf{B}}_{2j}}{s_{1\mathcal{S}j}s_{2\mathcal{S}j}}, ~~j=1,\ldots,p,
\end{align*}
where $s_{k\mathcal{S}j}^2 = \bm{e}_j^{\top}\left(\sum\nolimits_{i=1}^n\X_{k\mathcal{S}i}\X_{k\mathcal{S}i}^{\top}\right)^{-1}\bm{e}_j$, $k=1,2$, with narrowed subset $\mathcal{S}$.
The statistics ${W}_j$ has similar properties to the proposed one in \eqref{rankstatistic}, which is (asymptotically) symmetric with mean zero for $j\in\mathcal{A}^c$ and is a large positive value for $j\in\mathcal{A}$ without imposing the relationship between $Y$ and $\X$. Therefore, we propose to choose a threshold ${L}_+$
\begin{align}\label{threshold+}
	{L}_+ = \inf\left\{t >0: \frac{1+\#\left\{j: {W}_j\leq -t\right\}}{\#\left\{j:{W}_j\geq t\right\}\vee 1}\leq \alpha\right\},
\end{align}
and select the active variables by $\widehat{\mathcal{A}}({L}_+)=\left\{j: {W}_j \geq {L}_+, 1\leq j \leq p\right\}$ in high-dimensional setting.
The proposed ${L}_+$ in \eqref{threshold+} shares a similar spirit to Model-X Knockoff \citep{candes2018panning} or SDA filter \citep{du2021false} to obtain an accurate FDR control. 
However, in the high-dimensional variable selection problems, we usually can not collect enough information on $(Y,\X)$, and the exact knockoff copies may not be available when $p>n$. Fortunately, MFSDS does not require any prior information on the distribution of $(Y,\X)$ or the asymptotic distribution of statistics, and thus it is more suitable for high-dimensional problems.

\section{Theoretical results}\label{theory}

In this section, we entirely focus on controlling FDR. We begin by imposing a mild restriction on the response transformation function $f(Y)$. 

\begin{assump}[Response transformation]\label{ResTrans}
	Function $f(Y)$ satisfies $\mathbb{E}\left\{f(Y)\right\}=0$ and ${\rm{var}}\left\{f(Y)-\X^\top\bm\beta_0\right\}< \infty$.
\end{assump}

Assumption \ref{ResTrans} distinguishes our approach from most model-based selection methods by transforming a general model into a multivariate response linear problem, thereby achieving model-free variable selection \citep{wu2011asymptotic}. The transformed errors are not independent of the covariates and thus we need more effort for the theoretical analysis. Our first theorem is a finite sample theory for FDR control.

\begin{The}[Finite-sample FDR control]\label{exactFDR}
	Suppose Assumption \ref{ResTrans} hold. Assume that the statistics $W_j$, $1\leq j\leq p$, are well-defined. For any $\alpha\in (0,1)$, the FDR of our model-free selection procedure  satisfies 
	\begin{align*}
		{\rm{FDR}} \leq \min_{\epsilon \geq 0} \left\{\alpha\left(1+4\epsilon\right)+\text {Pr}\left(\max_{j\in\mathcal{A}^c}\Delta_j>\epsilon\right)\right\},
 	\end{align*}
	where  
	$\Delta_j=\left|{\text {Pr}}\left(W_j>0\mid |W_j|,{\bf{W}}_{-j}\right)-1/2\right|$ and ${\bf{W}}_{-j}=\left(W_1,\ldots,W_p\right)^\top \setminus W_j$.
\end{The}

Theorem \ref{exactFDR} holds  regardless of the unknown relationship between variables $\X$ and response $Y$. 
This result can be established using the techniques developed in \citet{barber2020robust}. 
The quantity $\Delta_j$ is interpreted as a measure to investigate the effect of both the asymmetry of $W_j$ and the dependence between $W_j$ and $\textbf{W}_{-j}$ on FDR. 
In asymmetric cases, it is still expected that $\Delta_j$ will be small, given that both $\widehat\beta^{(1)}_{hj}$ and $\widehat\beta^{(2)}_{hj}$ converge to normal distributions if $n$ is not too small.  Theorem \ref{exactFDR} implies that tight control of $\Delta_j$'s under asymmetric cases also results in effective FDR control.

For asymptotic FDR control of the proposed procedure, we require the following technical assumptions, which are not the weakest one but facilitate the technical proofs in the Supplementary Material. Let $\|{\bf A}-\bm\Sigma_{\mathcal{S}}\|_{\infty}=O_p(a_{np})$ with $a_{np}\to 0$,
where ${\bf A}=n^{-1}\sum\nolimits_{i=1}^n\X_{2\mathcal{S}i}\X_{2\mathcal{S}i}^{\top}$ and $\bm\Sigma_{\mathcal{S}}=\E(\X_{\mathcal{S}}\X_{\mathcal{S}}^{\top})$. Define $\upsilon_n=\max\left\{\|\bm\Sigma_{\mathcal{S}}\|_{\infty},\|\bm\Sigma^{-1}_{\mathcal{S}}\|_{\infty}\right\}$ 
and $\textbf{B}_{0\mathcal{S}} = \left\{\textbf{B}_{j}: j\in\mathcal{S}\right\}$. 
Denote $d_n=|\mathcal{A}|$, $q_n=|\mathcal{S}|$ and $q_{0n}=|\mathcal{S}\cap\mathcal{A}^c|$. Assume that $q_n$ is uniformly bounded above by some non-random sequence $\bar{q}_n$.

\begin{assump}[Sure screening property]\label{SureScreen}
	As $n\to\infty$, $\rm{Pr}(\mathcal{A}\subseteq\mathcal{S})\to 1.$
\end{assump}

\begin{assump}[Moments]\label{moment}
	Let $\bm\varepsilon={\bf f}-{\bf B}_{0\mathcal{S}}^\top\X_{\mathcal{S}} \in\mathbb{R}^H$. Conditioning on $\mathcal{S}$,
	there exists a positive diverging sequence $K_{n}$ 
	and a constant $\varpi>2$ such that
	\begin{align*}
		\max_{1\leq h\leq H}\max_{1\leq i\leq n}\mathbb{E}(\|\bm\Sigma_{\mathcal{S}}^{-1}\X_{2\mathcal{S}i}\varepsilon_{ih}\|_{\infty}^{\varpi})\leq K_{n}^{\varpi}, 
	\end{align*}
	for $i\in\mathcal{D}_2$. Assume that as $n\to \infty$, $\bar{q}_n^{1/\varpi+\gamma+1/2}K_{n}/n^{1/2-\gamma-1/\varpi}\to 0$ for some small $\gamma>0$.
\end{assump}

\begin{assump}[Design matrix]\label{DesignMatrix}
	There exist positive constants  $\underline{\kappa}$ and $\bar{\kappa}$ such that
	\begin{align*}
		\underline{\kappa}\leq \lim\inf_{n\rightarrow\infty}\lambda_{\min}( \X_{2\mathcal{S}}^\top\X_{2\mathcal{S}}/n)<\lim\sup_{n\rightarrow\infty}\lambda_{\max}( \X_{2\mathcal{S}}^\top\X_{2\mathcal{S}}/n)\leq\bar{\kappa},
	\end{align*}
	hold with probability one. 
\end{assump}

\begin{assump}[Estimation accuracy]\label{Estimationacc}
	Assume that $\|\widehat{\bf B}_{1j}-{\bf B}_{j}\|_2=O_p(c_{np})$ uniformly holds for $j\in\mathcal{S}$, where
	$\widehat{\bf B}_1$ is an estimator of $\bf B$ from $\mathcal{D}_1$ , $c_{np}\to 0$ and $1/(\sqrt{n}c_{np})=O(1)$.
\end{assump}

\begin{assump}[Signals]\label{signal}
	Denote $\mathcal{C}_{\bf B}=\left\{j\in\mathcal{A}: \|{\bf B}_j\|_2^2/\max\left\{c_{np}^2,\bar{q}_n\log {\bar{q}_n}/n\right\}\to\infty\right\}$. Let $\eta_{n}:= |\mathcal{C}_{\bf B}|\to\infty$ as $(n,p)\to\infty$.
\end{assump}

\begin{assump}[Dependence]\label{dependence}
	Let $\rho_{jl}$ denotes the conditional correlation between $W_{j}$ and $W_{l}$ given $\mathcal{D}_1$.
	Assume that for each $j$ and some $C>0$, $\#\{l\in\mathcal{A}^c: |\rho_{jl}|\geq C(\log n)^{-2-\nu}\}\leq r_p$, where $\nu>0$ is some small constant, and $r_p/\eta_{n}\to 0$ as $(n,p)\to\infty$.
\end{assump}

\begin{remark}\label{remark}
	Assumption \ref{SureScreen} has been used in \citet{meinshausen2009p,barber2019knockoff,du2021false} to ensure that $\widehat{{\bf B}}_{2j}$ is unbiased for $j\in\mathcal{S}$. Assumptions \ref{moment} and \ref{DesignMatrix} are commonly used in the context of variable selection. The rate $c_{np}$ in Assumption \ref{Estimationacc} ensures that $\widehat{\bf B}_1$ is a reasonable estimator of $\bf B$ from $\mathcal{D}_1$. For the LASSO selector, $c_{np}=d_n\sqrt{\log p/n}$ typically satisfies the Assumption \ref{Estimationacc}. Assumption \ref{signal} implies that the number of informative covariates with identifiable effect sizes is not too small as $(n,p)\to\infty$. Assumption \ref{dependence} allows $W_j$ to be correlated with all others but requires that the correlation coefficients need to converge to zero at a log rate. This condition is similar to the weak dependence structure given in \citet{fan2012estimating}.
\end{remark}


\begin{The}[Asymptotic FDR control]\label{AsyFDR}
	Suppose Assumptions \ref{ResTrans}$-$\ref{dependence} and hold. 
	For any $\alpha\in (0,1)$, {$c_{np}a_{np}\upsilon_n\sqrt{n\bar{q}_{n}}(\log{\bar{q}_n})^{3/2+\gamma}\to 0$ for a small $\gamma>0$,} the FDP of the MFSDS procedure with threshold $L$ satisfies
	\begin{align*}
		{\rm FDP}(L)&:=\frac{\#\{j: W_{j}\geq L,j \in\mathcal{A}^c\}}{\#\{j: W_{j}\geq L\}\vee 1}\leq\alpha+o_p(1), \label{fa}
	\end{align*}
	and $\mathop{\lim\sup}_{(n,p)\to\infty}{\rm FDR}\leq \alpha$.
\end{The}

Theorem \ref{AsyFDR} implies that the variable selection procedure with the data-driven threshold $L$ can control the FDR at the target level asymptotically. Further investigations are needed to better understand the condition $c_{np}a_{np}\upsilon_n\sqrt{n\bar{q}_{n}}(\log{\bar{q}_n})^{3/2+\gamma}\to 0$. The conventional result of $\|{\bf A}-\bm\Sigma_{\mathcal{S}}\|_{\infty}$ indicates that $a_{np}=O_p(\upsilon_n\sqrt{\log \bar{q}_n/n})$. With $c_{np}=d_n\sqrt{\log p/n}$ of LASSO selector, the condition degenerates to $d_n\upsilon_n^2\sqrt{\bar{q}_n/n}\to 0$ if $p$ is of a polynomial rate of $n$. The above condition basically imposes restrictions on the rate of $d_n$, $\upsilon_n$, and $\bar{q}_n$. Accordingly, the screening stage on split $\mathcal{D}_1$ must satisfy $\bar{q}_n=o(n)$ if we assume that $d_n$ and  $\upsilon_n$ are bounded. Alternatively, if we only assume that $\upsilon_n$ is bounded, then a sufficient requirement for the condition in Theorem \ref{AsyFDR} is $\bar{q}_n=o(n^{1/2})$ since $d_n\leq \bar{q}_n$. This is a reasonable rate in the problem with a diverging number of parameters, such as \citet{fan2004nonconcave} and \citet{wu2011asymptotic}.


\section{Numerical studies}\label{numerical}

We evaluate the performance of our proposed procedure on several simulated datasets and a real-data example under low-dimensional and high-dimensional settings.

\subsection{Implantation details}

We compared our MFSDS with several benchmark methods. The first one is the marginal coordinate test in sliced inverse regression \citep[SIR,][]{cook2004testing}, which aims at controlling the error rate for each coordinate. To make a global error rate control,  we then apply the BH procedure \citep{benjamini1995controlling} to the $p$-values. This method is implemented using functions ``dr" and ``drop1" in \texttt{R} package \texttt{dr}. The second method is the Model-X Knockoff \citep{candes2018panning}, which also is a model-free and data-driven variable selection procedure as the proposed method. This method is implemented by the function ``create.gaussian'' in \texttt{R} package \texttt{knockoff} using the lasso signed maximum feature important statistics. 
The two methods are termed MSIR-BH and MX-Knockoff, respectively.

We set the FDR level to $\alpha=0.2$ and conducted 500 replications for all simulation results. The performance of the proposed MFSDS is evaluated along with the above benchmarks through the comparisons of FDR, the true positive rate (TPR), $P_a=\text{Pr}(\mathcal{A}\subseteq\widehat{\mathcal{A}}(L))$ and the average computing time.


\subsubsection{Low-dimensional studies}\label{Sec:LowdimResult}

We generate the covariates $\X$ following three distributions: multivariate normal distribution $\mathcal{N}\left(0,\bm\Sigma\right)$ with $\bm\Sigma=(\sigma_{ij})=\rho^{|i-j|}$, $1\leq i,j\leq p$; multivariate $t(5)$ distribution with covariance $\bm\Sigma$; a mixed distribution which consists of $\left\{X_j\right\}_{j=1}^{[p/3]}$ are from $\mathcal{N}\left(0,\bm\Sigma_{[p/3]}\right)$, $\left\{X_j\right\}_{j=[p/3]+1}^{[2p/3]}$ are from $\mathcal{N}\left(0,{\bf I}_{[p/3]}\right)$, and $\left\{X_j\right\}_{j=[2p/3]+1}^{p}$ are i.i.d from a $t(5)$ distribution. The error term $\eta$ is the standard normal distribution which is independent of $\X$. We fix $\left(p,p_1\right)=(20,10)$. For a scalar $c$, write $\bm c_p = (c,\dots,c)$ be the $p$-dimensional row vector of $c$'s. Five models have been considered:

\begin{itemize}
	\item \textbf{Scenario 1a}: 
	$Y = \bm\beta^{\top}\X+3\eta$,
	where  $\bm\beta = (\bm 1_{p_1},\bm 0_{p-p_1})^\top$. 
	\item \textbf{Scenario 1b}: 
	$Y=\left|\bm{\beta}_1^{\top}\X\right|+\exp\left(3+\bm{\beta}_2^{\top}\X\right) +\eta$,
	where $\bm\beta_1 = (\bm 1_5,\bm 0_{p-5})^{\top}$,  $\bm\beta_2 = (\bm 0_5,\bm 1_5, \bm0_{p-p_1})^{\top}$. 
	\item \textbf{Scenario 1c}: 
	$Y=\bm{\beta}_1^{\top}\X+\left(\bm{\beta}_2^{\top}\X+3\right)^2+\exp\left(\bm{\beta}_3^{\top}\X\right)+\eta$,
	where $\bm\beta_1 = (\bm 1_3,\bm 0_{p-3})^{\top}$, $\bm\beta_2 = (\bm 0_3,\bm 1_3,\bm 0_{p-6})^{\top}$, $\bm\beta_3 = (\bm 0_6,\bm 1_4,\bm 0_{p-p_1})^{\top}$. 
	
	\item \textbf{Scenario 1d}: $Y=\bm{\beta}_1^\top \X/\left\{0.5+(1.5+\bm{\beta}_2^\top\X)^2\right\}+(\bm{\beta}_3^\top\X)^2+\exp(\bm{\beta}_4^\top\X)+\eta$, where $\bm\beta_1 = (\bm 1_2,\bm 0_{p-2})^{\top}$, $\bm\beta_2 = (\bm 0_2,\bm 1_{2},\bm 0_{p-4})^{\top}$, $\bm\beta_3 = (\bm 0_4,\bm 1_{2},\bm 0_{p-6})^{\top}$, $\bm\beta_4 = (\bm 0_6,\bm 1_4,\bm 0_{p-p_1})^{\top}$.
	
\item \textbf{Scenario 1e}: $Y=2(\bm{\beta}_1^\top \X)^2\sin(\bm{\beta}_2^\top \X)+3(\bm{\beta}_3^\top \X)^3\exp(\bm{\beta}_4^\top \X)+|{\bm{\beta}_5^\top \X}|\eta$, where $\bm\beta_1 = (\bm 1_2,\bm 0_{p-2})^{\top}$, $\bm\beta_2 = (\bm 0_2,\bm 1_{2},\bm 0_{p-4})^{\top}$, $\bm\beta_3 = (\bm 0_4,\bm 1_{2},\bm 0_{p-6})^{\top}$, $\bm\beta_4 = (\bm 0_6,\bm 1_2,\bm 0_{p-8})^{\top}$, $\bm\beta_5 = (\bm 0_8,\bm 1_2,\bm 0_{p-p_1})^{\top}$.
	
\end{itemize}

We first consider three response transformation functions $f_h(Y), h=1,\ldots, H$, for the proposed MFSDS: (1) the slice indicator function 
\citep{li1991sliced} that $f_h(Y)=1$ if $Y$ is in the $h$th slice and 0 otherwise; (2) The CIRE-type response transformation \citep{cook2006using} 
that $f_h(Y)=Y$ if $Y$ is in the $h$th slice and 0 otherwise; (3) the normalized polynomial response transformation \citep{yin2002dimension} that $f_h(Y)=Y^h$ if $Y$ is in the $h$th slice and 0 otherwise. We name these three functions as Indicator, CIRE, and Poly, respectively.

Figure \ref{Fig:H-FDR} shows that our proposed procedure successfully controls FDR in an acceptable range of the target level, regardless of the number of working dimension $H$ and the response transformation functions. The three response transformation functions exhibit similar patterns with FDR control, and we do not address which $f(Y)$ is the ``best'' in this paper. 
Our methodology does not require the estimation of $d$ with a given $H$. In the rest of the simulations, we focus on the slice indicator function and fix $H=4$ for the proposed MFSDS.

\begin{figure}[htbp]
	\centering
	\includegraphics[scale=0.5]{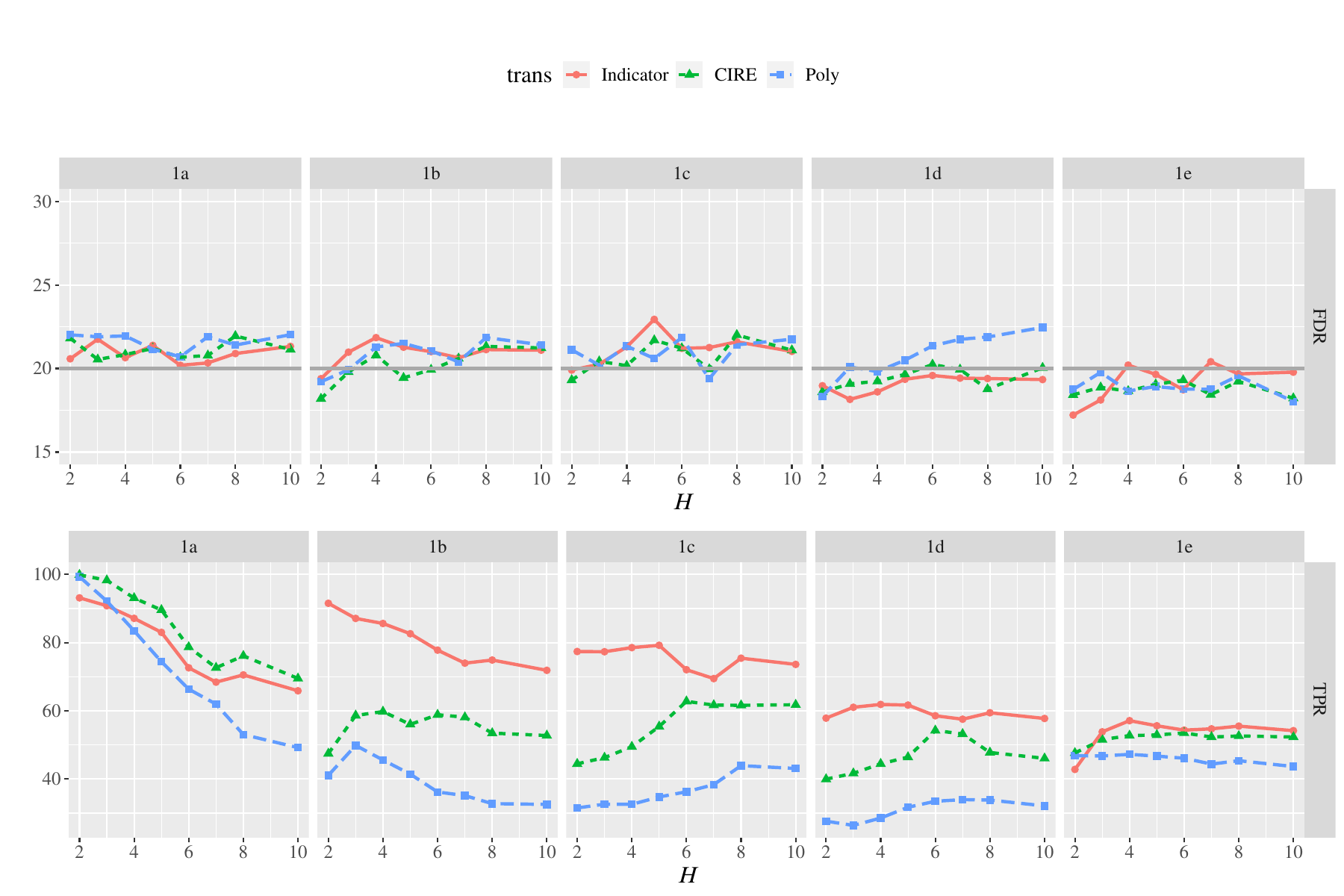} 
	\caption{FDR and TPR(\%) of the proposed MFSDS against different $H$ and $f_h(Y)$ under Scenarios 1a$-$1e when $(n, \rho)=(500, 0.5)$  and $\X$ is normal distribution. The gray solid line denotes the target FDR level.
	}
	\label{Fig:H-FDR}
\end{figure}

Next, we compare FDR and TPR of the proposed MFSDS under low-dimensional settings with marginal SIR and MX-Knockoff in Table \ref{lowdim-Xdistribution} and Table \ref{lowdim-Xcorrelation-normal}. Table \ref{lowdim-Xdistribution} studies how the proposed MFSDS and the benchmark methods are affected by the covariate distributions. 
Table \ref{lowdim-Xcorrelation-normal} displays the comparisons of covariate correlation for the three methods. Across all scenarios, the FDRs of MFSDS persist at the desired level consistently and the TPRs of MFSDS are higher than MSIR-BH and MX-Knockoff in most cases. Further discussion can be found below. 	
\begin{itemize}
	
	\item [(a)] \emph{MFSDS vs MSIR-BH}.
	The marginal SIR \citep{cook2004testing} with BH procedure \citep{benjamini1995controlling} controls the FDR in all settings but can be overly  conservative in some cases. This is because the BH procedure controls the FDR at level $\alpha p_0/p$ in a low-dimensional setting. By contrast, MFSDS performs accurate FDR control and better TPR in both linear and nonlinear models. It seems that the power of the proposed MFSDS is slightly lower than MSIR-BH in very few cases since MFSDS involves data splitting to construct the statistic symmetry property. But additional simulations show that, as the sample size $n$ increases, our method will be more effective than others in low-dimensional settings.
	
	\item [(b)] \emph{MFSDS vs MX-Knockoff}. MFSDS and MX-Knockoff are both model-free and data-driven variable selection methods. However, MX-Knockoff requires knowing the joint distribution of the covariates. The proposed MFSDS controls the FDR more accurately near the desired level, 
	while MX-Knockoff fails to control the FDR for $t(5)$ distribution in Column 2 of Table \ref{lowdim-Xdistribution}.  Furthermore, Table \ref{lowdim-Xdistribution} shows that MX-Knockoff works not well as our MFSDS in nonlinear cases. Table \ref{lowdim-Xcorrelation-normal}  indicates that the MFSDS robustly controls the FDR at the nominal level, but the MX-Knockoff exhibits more conservative FDR and lower TPR as the correlation increases across all scenarios.

\end{itemize}

\begin{table}[htbp]
	\caption{FDR and TPR (\%) for three methods against different covariate distributions under Scenarios 1a$-$1e when $(n,\rho)=(500, 0.5)$.}  
	\vskip .2cm
	\label{lowdim-Xdistribution}
	\centering
	{\centering
		\scalebox{1}{
			\begin{tabular}{clrrrrrrrrrrrrrrrr}
				\toprule[1pt]
				& 	&  \multicolumn{2}{c}{{\bf normal}} &&\multicolumn{2}{c}{\bf{t5}}
				&&\multicolumn{2}{c}{\bf{mixed}}\\
				\cline{3-4} \cline{6-7} \cline{9-10}
				\bf{Scenario}&\bf{Method}& FDR&TPR&&FDR&TPR&&FDR&TPR\\
				
				\hline
				& MFSDS      &20.7&87.1&&22.0&84.4&&22.2&96.4\\
				1a& MSIR-BH   &10.4&91.4&&10.1&84.6&&10.8&98.2\\
				& MX-Knockoff&21.0&99.9&&22.4&100.0&&22.8&100.0\\
				
				\cline{2-10}
				& MFSDS      &21.8&85.6&&19.1&70.4&&21.5&89.7\\
				1b& MSIR-BH&10.0&83.1&&13.1&68.3&&10.5&88.9\\
				& MX-Knockoff&13.3&39.9&&39.4&50.6&&16.3&58.9\\
				
				\cline{2-10}
				& MFSDS      &21.3&78.5&&20.9&68.9&&20.9&84.2\\
				
				1c& MSIR-BH   &10.5&74.8&& 14.1&68.1 &&11.2&80.9\\
				& MX-Knockoff&13.6&46.0&&38.8&52.4&&19.1&82.2\\
				
				\cline{2-10}
				& MFSDS      &18.6&61.9&&20.0&51.9&&20.0&61.9\\
				1d& MSIR-BH  &10.6&57.4&&14.1&48.5&&11.6&56.0\\
				& MX-Knockoff&14.7&40.2&&38.9&52.5&&18.6&51.9\\
				
				\cline{2-10}
				& MFSDS      &20.2&57.1&&18.7&49.5&&18.0&58.2\\
				1e& MSIR-BH  &11.7&51.1&&12.9&46.2&&10.9&52.8\\
				& MX-Knockoff&14.5&40.0&&38.8&52.6&&16.0&33.4\\
				
				\bottomrule[1pt]
	\end{tabular}}}\\
\end{table}

\begin{table}[t!]
	\caption{FDR and TPR (\%) for three methods against different correlation $\rho$ under Scenarios 1a$-$1e when $n=500$ and $\X$ is from normal distribution.}
	\vskip .2cm
	\label{lowdim-Xcorrelation-normal}
	\centering
	{\centering
		\scalebox{1}{
			\begin{tabular}{clrrrrrrrrrrrrrrrr}
				\toprule[1pt]
				& 	&  \multicolumn{2}{c}{{$\rho=0.2$}} &&\multicolumn{2}{c}{{$\rho=0.5$}}
				&&\multicolumn{2}{c}{{$\rho=0.8$}}\\
				\cline{3-4} \cline{6-7} \cline{9-10}
				{\bf Scenario}&{\bf Method}& FDR&TPR&&FDR&TPR&&FDR&TPR\\
				
				\hline
				& MFSDS      &22.6&99.1&&20.7&87.1&&14.6&27.7 \\
				1a& MSIR-BH   &11.0&99.9&&10.4&91.4&&14.0&10.0 \\
				& MX-Knockoff &22.8&100.0&&21.0&99.9&&20.6&95.6\\
				
				\hline
				& MFSDS     &22.9&99.5&&21.8&85.6&&21.6&30.2 \\
				1b& MSIR-BH   &11.2&99.8&&10.0&83.1&&12.0&16.5 \\
				& MX-Knockoff &15.9&54.9&&13.3&39.9&&10.1&20.7\\
				
				\hline
				& MFSDS      &21.6&89.9&&21.3&78.5&&20.2&34.9 \\
				1c& MSIR-BH   &11.2&85.9&&10.5&74.8&&11.7&18.4 \\
				& MX-Knockoff &17.5&69.3&&13.6&46.0&&11.1&22.6\\
				
				\hline
				& MFSDS           &18.2&64.8&&18.6&61.9&&20.0&40.9\\
				1d& MSIR-BH    &10.2&61.6&&10.6&57.4&&12.5&31.2\\
				& MX-Knockoff &16.7&48.0&&14.7&40.2&&11.5&22.5\\
				
				\hline
				& MFSDS           &19.1&64.7&&20.2&57.1&&18.1&45.6\\
				1e& MSIR-BH    &11.2&61.2&&11.7&51.1&&9.6&40.2\\
				& MX-Knockoff&17.0&38.3&&14.5&40.0&&10.6&25.1\\
				
				\bottomrule[1pt]
	\end{tabular}}}\\
\end{table}

\subsubsection{High-dimensional studies}\label{High-dimensional studies}

In high-dimensional settings, we consider the following benchmarks. The competitor MSIR-BH in low-dimension does not work in high-dimensional settings since the $p$-value cannot directly be obtained. Thus, the sample-splitting method \citep{wasserman2009high}, which first conducts data screening using LASSO and then applies BH to the $p$-values calculated by marginal SIR \citep{cook2004testing}. Since the commonly used Akaike information criterion such as in  \citet{du2021false} causes inaccurate model deviance after slicing responses, the cross-validation criterion is conducted to choose an overfitted model in the screening stage. The second method is named as MFSDS-DB
\citep{javanmard2014confidence}, which extends the least square solution in \eqref{OLS} to regularized version with $L_1$ penalty in \eqref{sparseOLS} and makes a bias correction with \texttt{R} package \texttt{selectiveInference}.  The MX-Knockoff is conducted by the function \texttt{create.second$_-$order} in \texttt{R} package \texttt{knockoff} to approximate an accurate precision matrix in high-dimensional setting \citep{candes2018panning}. 
The fourth one is the marginal independence SIR proposed in \citet{yu2016marginal}. We choose two model sizes $\lfloor cn/\log(n)\rfloor$, $c=(0.05,0.5)$, as two simple competitors and denote them as IM-SIR1 and IM-SIR2, respectively. We consider the following models when $p>n$ with signal strength $a$.
\begin{itemize}
	\item \textbf{Scenario 2a}: 
	$Y = a\cdot\exp\left(5+\bm\beta^{\top}\X\right)+\eta$,
	where $\bm\beta = (\bm 1_{p_1},\bm 0_{p-p_1})^\top$. 
	\item \textbf{Scenario 2b}: 
	$Y=a\cdot\left\{2\bm{\beta}_1^{\top}\X+3\exp\left(\bm{\beta}_2^{\top}\X\right) \right\}+\eta$,
	where $\bm\beta_1 = (\bm 1_5, \bm 0_{p-5})^{\top}$,  $\bm\beta_2 = (\bm 0_5, \bm 1_5, \bm 0_{p-p_1})^{\top}$. 
	\item \textbf{Scenario 2c}:  
	$Y=a\cdot\left\{\bm{\beta}_1^{\top}\X+\left|\bm{\beta}_2^{\top}\X+5\right|+\exp\left(\bm{\beta}_3^{\top}\X\right)\right\}+\eta$,
	where $\bm\beta_1 = (\bm 1_3,\bm 0_{p-3})^{\top}$; $\bm\beta_2 = (\bm 0_3, \bm 1_3, \bm 0_{p-6})^{\top}$; $\bm\beta_3 = (\bm 0_6,\bm 1_4,\bm 0_{p-p_1})^{\top}$.
\end{itemize}

\begin{table}[htbp]
	\caption{FDR, TPR, $P_a$(\%)
		and computing time (in second) for several methods against different $\X$ distributions under Scenarios 2a$-$2c when $(n,p,p_1,\rho, a)=(500,1000,10,0.5,1)$.}
	\vspace{0.2cm}
	\label{highdimension-Xd}
	\centering
	{\centering
		\scalebox{1}{
			\begin{tabular}{clrrrrrrrrrrrrrrrr}
				\toprule[1pt]
				&	&  \multicolumn{4}{c}{\textbf{normal}} &&\multicolumn{4}{c}{\textbf{mixed}}\\
				\cline{3-6} \cline{8-11}
				\textbf{Scenario} &\textbf{Method}& FDR&TPR&$P_a$&time&&FDR&TPR&$P_a$&time\\
				
				\hline
				&MFSDS&18.3&98.7&90.2&12.9&&17.5&98.3&86.6&14.1\\
				&MFSDS-DB&17.1&57.9&7.0&75.1&&17.0&77.3&19.8&65.5\\
				&MSIR-BH&4.5&32.8&0.0&11.4&&4.2&32.9&0.0&12.6\\
				2a & IM-SIR1&0.0&40.0&0.0&26.3&&0.0&40.0&0.0&34.8\\
				&  IM-SIR2&75.0&100.0&100.0&26.3&&75.0&100.0&100.0&34.5\\
				& MX-Knockoff&6.5&4.3&0.0&31.8&&26.4&9.2&0.0&31.6\\
				\hline
				&MFSDS&17.0&94.7&62.0&12.8&&17.5&94.6&58.8&14.7\\
				&MFSDS-DB&15.6&71.4&4.2&74.8&&16.9&80.5&14.4&75.3\\
				&MSIR-BH&5.3&34.9&0.0&11.5&&4.6&34.6&0.0&13.3\\
				2b & IM-SIR1&0.0&40.0&0.0&26.3&&0.0&40.0&0.0&27.3\\
				&  IM-SIR2&75.0&100.0&100.0&26.3&&75.0&100.0&100.0&27.1\\
				& MX-Knockoff&10.3&11.6&0.0&31.9&&29.0&19.8&0.0&31.4\\
				
				\hline
				&MFSDS&17.9&92.7&50.6&12.4&&19.2&92.8&49.8&22.8\\
				&MFSDS-DB&17.0&62.1&8.0&102.3&&17.4&73.2&6.0&126.8\\
				&MSIR-BH&6.5&33.4&0.0&12.2&&6.3&31.1&0.0&21.8\\
				2c & IM-SIR1&0.0&40.0&0.0&22.7&&0.0&40.0&0.0&36.2\\
				&  IM-SIR2&75.0&100.0&100.0&22.9&&75.0&100.0&100.0&35.2\\
				& MX-Knockoff&12.1&12.2&0.0&28.4&&28.8&19.1&0.0&42.0\\
				\bottomrule[1pt]
	\end{tabular}}}\\
\end{table}

Table \ref{highdimension-Xd} presents the comparison results for different covariate distributions to investigate the error rate control and detection power under high-dimensional settings. The FDRs of the proposed MFSDS are approximately controlled at the target level $\alpha$ with a higher power, which is consistent with our theory. A similar analysis also can be found in the Supplementary Material Table \ref{highdimension-Xd-n800} with a larger sample size. 
Table \ref{highdimension-Xd} and Table \ref{highdimension-Xd-n800} further demonstrate that the MFSDS is able to detect all active variables when $n$ is large. As we can expect, although MFSDS involves data splitting which may lose some data information, its power can still be higher since the feature screening step significantly increases the signal-to-noise ratio. Besides, our method exhibits lower computing time since we avoid constructing the asymptotic distribution for each dimension in marginal coordinate test \citep{cook2004testing} and generating the knockoff copies in Model-X knockoff \citep{candes2018panning}.  We provide further explanations below.

\begin{itemize}
	\item [(a)] \emph{MFSDS vs MFSDS-DB}. 
	The FDR of the MFSDS-DB method controls pretty well as the proposed MFSDS but it performs a lower power than MFSDS since MFSDS-DB uses bias correction instead of the screening stage which may not boost the signal-to-noise ratio. We know that the MFSDS-DB method is an extension of our low-dimensional procedure with a debiased lasso estimate but it needs to estimate the precision matrix which results in significantly higher computational costs.

	\item [(b)] \emph{MFSDS vs MSIR-BH}. MSIR-BH method is a post-selection inference with marginal statistics, which achieves a conservative FDR control compared with MFSDS. It adopts data splitting \citep{meinshausen2009p} but they only construct the marginal test statistics on $\mathcal{D}_2$ which suffers from a serious power loss. As pointed out by reviewers,  \citet{guo2024model} provided a nice refinement of MSIR-BH without data splitting, but the decorrelating process is complicated similar to the debiasing in MFSDS-DB; additional simulation in Table \ref{Tab:Guo} shows that it leads intensive computation.
	
	\item [(c)] \emph{MFSDS vs IM-SIR}. The hard thresholding IM-SIR methods can detect more active variables only when model size $\lfloor cn/\log(n)\rfloor$ is greater than $p_1$. Table \ref{highdimension-Xd} implies that the hard-thresholding method can not control the FDR, and thus their large powers are unreliable with user-specified model sizes.
	
	\item [(d)] \emph{MFSDS vs MX-Knockoff}. MX-Knockoff offers a variable selection solution without making any modeling assumptions in high-dimensional situations.
	One important assumption in the theoretical development of MX-Knockoff is that the joint distribution of covariates $\X$ should be either exactly known or should be estimated robustly. Moreover, MX-Knockoff usually considers Gaussian distribution or uses a second-order approximate construction, which may lead to some invalid FDR control and power loss. By contrast, MFSDS controls the FDR more accurately under mixed distribution. In addition, the symmetry property of MFSDS stems from sample splitting while MX-Knockoff requires variable augmentation, which results in a large amount of calculation.
	
\end{itemize}

\begin{figure}[htbp]
	\centering
	\includegraphics[scale=0.48]{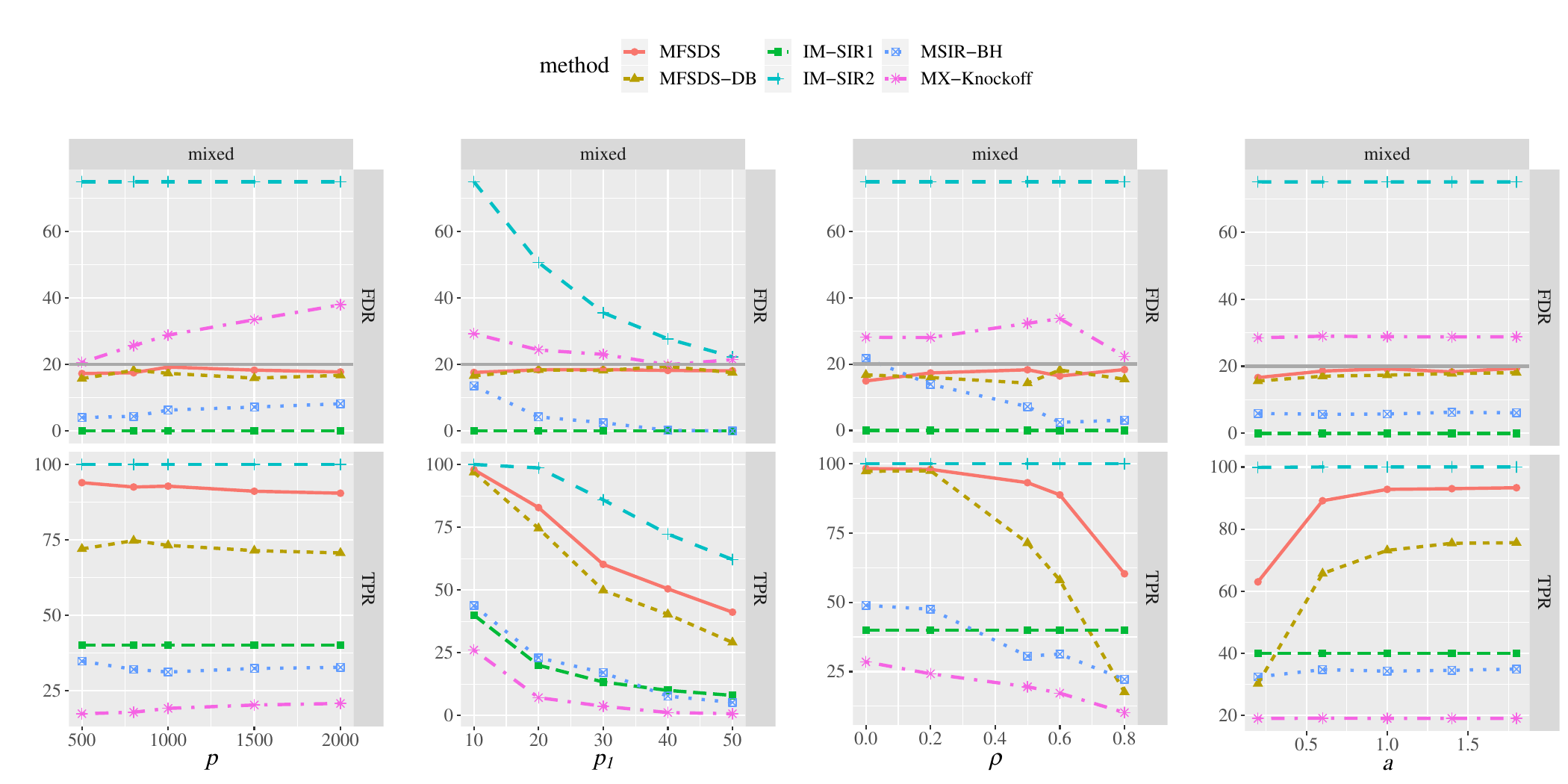}
	\caption{ FDR and TPR (\%) curves against different covariate dimension $p$, signal number $p_1$,
		correlation $\rho$,  and signal strength $a$ under Scenario 2c when $n=500$ and $\X$ is from mixed distribution. 
		The gray solid line denotes the target FDR level.
	}
	\label{Fig:highdim-mixed}
\end{figure}

To further investigate the efficiency of our MFSDS procedure in high-dimensional settings against different covariate dimension $p$,  the signal number $p_1$, covariate correlation $\rho$,  and signal strength $a$, we report the corresponding FDR and TPR in Figure \ref{Fig:highdim-mixed}. The FDR of MFSDS remains robust in an acceptable range of the target level no matter the dimension, signal number, correlation, and signal strength varied. Our MFSDS consistently achieves the most powerful TPR than other competitors except for IM-SIR2 since IM-SIR2 can not make a fair error rate control. The practical performance between MFSDS and MX-Knockoff is quite different. It is clear that although controlling the FDR below the target level with larger $p_1$ in Figure \ref{Fig:highdim-mixed}, MX-Knockoff suffers from a larger power loss as the signal number $p_1$ and correlation $\rho$ increases. 
Additional numerical results with normal covariate distribution can be found in the Supplementary Material Figure \ref{highdim-prhom-norm}. Figure \ref{highdim-prhom-norm} shows that the FDR of MX-Knockoff controls quite well as MFSDS under normal distribution with various combinations $p$, $p_1$, $\rho$, and $a$, 
but still a significant power gap compared to our proposed MFSDS.

\subsection{Real data implementation}\label{realdata}

In this section, we apply our proposed MFSDS procedure
to the children cancer data for classifying small round blue cell tumors (SRBCT), which has been analyzed by \citet{khan2001classification} and \citet{yu2016marginal}. The SRBCT dataset aims to classify four classes of different childhood tumors sharing similar visual features during routine histology. Data were collected from 83 tumor samples and the expression measurements on 2308 genes for each sample are provided. In this dataset, we focus on investigating the performance of the FDR control between our MFSDS procedure and other existing methods. 

\begin{figure}[htbp]
	\centering
	\includegraphics[height=10cm,width=16cm]{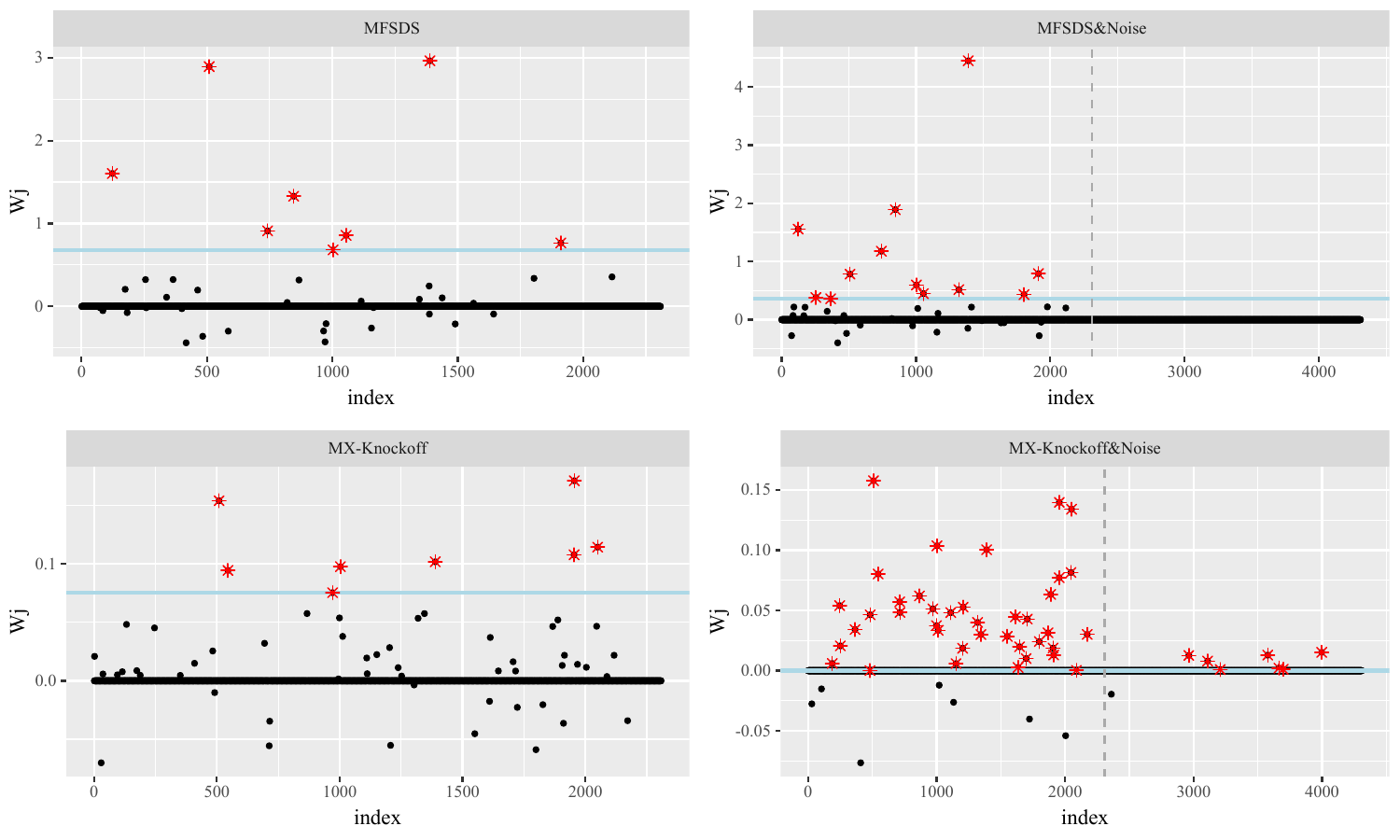}
	\caption{ Scatterplots of $W_j$'s for MFSDS and MX-Knockoff with the red points denoting selected variables and lightblue line denoting the data-driven threshold when $\alpha=0.2$.
	}
	\label{scatterplot-realdata}
\end{figure}

We first present the scatterplots of the ranking statistics $W_j$'s for MFSDS and MX-Knockoff in Figure \ref{scatterplot-realdata}. As we expect, genes with larger $W_j$'s are possible to be selected as influential genes (red dots), while the unselected genes (black dots) are roughly and symmetrically distributed around zero line, even though many $W_j$'s are exactly zero.

\begin{table}[htbp]
	\caption{Number of discovered (ND) genes and computing time (in second) in SRBCT and SRBCT\&Noise datasets. Values in parentheses are the number of discovered noise variables.}
	\vspace{0.2cm}
	\label{ND-realdata}
	\centering
	{\centering
		\scalebox{1}{
			\begin{tabular}{lrrrrrrrrrrrrrrrr}
				\toprule[1pt]
				&  \multicolumn{2}{c}{SRBCT} &&\multicolumn{2}{c}{SRBCT\&Noise}\\
				\cline{2-3} \cline{5-6}
				Method& ND&time&&ND&time\\
				
				\hline
				MFSDS &8&1.52&&12(0)&1.92\\
				MSIR-BH &5&0.50&&5(0)&0.72\\
				IM-SIR(c=0.05) &1&130.20&&1(0)&1514.23 \\
				IM-SIR(c=1) &18&128.31&&18(0)&1516.59\\
				MX-Knockoff &8&49.82&&46(7)&240.20\\
				
				\bottomrule[1pt]
	\end{tabular}}}\\
\end{table}

Next, we introduce some simulated variables as inactive genes to further investigate the performance of our proposed method. Specifically, 1000 noise variables $Z_1$ are from $\mathcal{N}(0,1)$, and another 1000 noise variables $Z_2$ are from $t(3)$, which are all independently and randomly generated. We refer to this dataset as SRBCT\&Noise, and our main objective is to identify the important genes from the whole 4328 variables. Since $Z_1$ and $Z_2$ are actually inactive by construction, they can be viewed as truly false discoveries. The scatterplot of $W_j$'s and the number of discovery genes are presented in Figure \ref{scatterplot-realdata} and Table \ref{ND-realdata}. 
Yu et al. \cite{yu2016marginal} analyzed that an average of 9 genes out of the 2308 total genes were selected, which is very close to the number of variables selected by our method. We observed that MFSDS and MX-Knockoff identified the same number of important genes without noise variables, but MX-Knockoff mistakenly selected a small number of truly noise variables, i.e., the added noise variables. Besides, the proposed MFSDS demonstrates a relatively smaller computing time than IM-SIR and MX-Knockoff.

\section{Discussion}\label{discussion}

Identifying the truly  contributory variables is a critical task in statistical inference. In this paper, we proposed a novel model-free variable selection procedure with a data-driven threshold in sufficient dimension reduction framework for a family of sliced methods via data splitting. 
The proposed MFSDS is computationally efficient in high-dimensional settings. Theoretical and numerical results show that the proposed MFSDS can asymptotically control the FDR at the target level. 

Our work raises several intriguing open questions that deserve further investigation. First, 
it is interesting to adopt multiple data splitting in MFSDS to improve the stability and robustness but it may require intensive computations. Equal size data splitting is used in our simulation and theory for simplicity. Second, we have checked that a larger sample size for $\mathcal{D}_1$ will improve TPR to some extent. In addition, we sacrifice a little detection power to obtain symmetry property via data splitting, while how to construct symmetry property without data splitting or deriving complex asymptotic distribution may be another promising direction. Third, benefiting from the use of the response transformation, the proposed method can work with many classical SDR methods (such as sliced inverse regression \citep{li1991sliced}, covariance inverse regression estimator \citep{cook2006using}) and many types of response (such as continuous, discrete, or categorical data). Finally, our method may also be able to dispose of the order determination problem in many situations which merits further investigation, such as \citet{luo2016combining} if the importance of $H$ slices has been determined in advance with a diverging number of $H$.  It is also worth further exploring the optimal selection of working dimension $H$ in practice. 


\appendix
\section*{Appendix}

This Appendix gives some key lemmas and succinct proof of theorems. We only consider the proofs under a high-dimensional setting and the low-dimensional version can be easily verified. Additional lemmas used in Appendix with their proofs and some additional simulation results can be found in the Supplementary Material.


\renewcommand{\thesection}{}
\setcounter{Lem}{0}
\renewcommand{\theLem}{A.\arabic{Lem}}

\setcounter{Pro}{0}
\renewcommand{\thePro}{A.\arabic{Pro}}

\renewcommand{\theequation}{}
\setcounter{equation}{0}
\renewcommand{\theequation}{A.\arabic{equation}}

For any vector ${\bm v} \in \mathbb{R}^H$, we have
\begin{align*}
	\left(\widehat{\bf B}_{2j}-{\bf B}_j\right)^\top{\bm v} &= \bm{e}_j^\top \left(\frac{1}{n}\sum\limits_{i=1}^n\X_{2\mathcal{S}i}\X_{2\mathcal{S}i}^\top\right)^{-1}\frac{1}{n}\sum\limits_{i=1}^n\X_{2\mathcal{S}i}\left({\bf{f}}_{2i}-{\bf B}_0^\top\X_{2\mathcal{S}i}\right)^\top{\bm v}\\
	&=\bm{e}_j^\top\left(\frac{1}{n}\sum\limits_{i=1}^n\X_{2\mathcal{S}i}\X_{2\mathcal{S}i}^\top\right)^{-1}\frac{1}{n}\sum\limits_{i=1}^n\X_{2\mathcal{S}i}\bm{\varepsilon}_i^\top\bm{v}.
\end{align*}
Here $\bm\varepsilon=\left(\varepsilon_1,\ldots,\varepsilon_H\right)^\top={\bf f}-{\bf B}_{0\mathcal{S}}^\top\X_{\mathcal{S}}\in\mathbb{R}^{H}$ with $\mathbb{E}(\X_{\mathcal{S}}\bm\varepsilon^{\top})= 0$. However, $\bm\varepsilon$ is not independent with $\X_{\mathcal{S}}$ nor $\mathbb{E}\left(\bm\varepsilon\mid \X_{\mathcal{S}}\right)=0$ due to the use of response transformation, which is quite different from the usual setting. To establish the FDR control result of our proposed procedure, we construct a modified statistic
\begin{align*}
	\widetilde{W}_j = \frac{ \widehat{{\bf B}}_{1j}^\top \widetilde{{\bf B}}_{2j}}{s_{1\mathcal{S}j}\widetilde{s}_{2\mathcal{S}j}}, ~~ j=1,\ldots,p,
\end{align*}
where $\widetilde{\bf B}_{2j}=\bm e_j^\top n^{-1}\bm\Sigma_{\mathcal{S}}^{-1}\sum_{i=1}^n \X_{2\mathcal{S}i}\bm\varepsilon_{2i}^\top$, $\widetilde s^2_{2\mathcal{S}j}=n^{-1}\bm e_j^\top\bm\Sigma_{\mathcal{S}}^{-1}\bm e_j$ for $j\in\mathcal{S}$, and $\widetilde{W}_j=0$ otherwise. Let $G(t)={q_{0n}}^{-1}\sum_{j\in \mathcal{A}^c}{\Pr}(\widetilde W_j\geq t\mid \mathcal{D}_{1})$, $G_{-}(t)=q_{0n}^{-1}\sum_{j\in \mathcal{A}^c}\Pr(\widetilde W_j\leq-t\mid \mathcal{D}_{1})$ and $G^{-1}(y)=\inf\{t\geq 0: G(t)\leq y\}$ for $0\leq y\leq 1$.

Before we present the proofs of the main theorems, we first state two key lemmas. The first lemma characterizes the closeness between $G(t)$ and $G_{-}(t)$, which plays an important role in the proof.

\begin{Lem}[Symmetry property]\label{klem1}
	Suppose Assumptions \ref{ResTrans}$-$\ref{DesignMatrix} hold. For any $0 \leq t\leq G_{-}^{-1}(\alpha\eta_{n}/{q_{0n}})$
	\begin{align*}
		\frac{G(t)}{G_{-}(t)}-1\to 0.
	\end{align*}
\end{Lem}

The next lemma establishes the uniform convergence of ${\sum_{j\in \mathcal{A}^c}\mathbb{I}(\widetilde W_{j}\geq t)}/\left({q_{0n}G(t)}\right)$.

\begin{Lem}[Uniform consistency]\label{klem2}
	Suppose Assumptions \ref{ResTrans}$-$\ref{DesignMatrix} and \ref{dependence} hold. Then
	\begin{align}
		\sup_{0\leq t\leq G^{-1}(\alpha \eta_{n}/q_{0n})}\left|\frac{\sum_{j\in \mathcal{A}^c}\mathbb{I}(\widetilde W_{j}\geq t)}{q_{0n}G(t)}-1\right|&=o_p(1),\\
		\sup_{0\leq t\leq G_{-}^{-1}(\alpha \eta_{n}/q_{0n})}\left|\frac{\sum_{j\in \mathcal{A}^c}\mathbb{I}(\widetilde W_{j}\leq -t)}{q_{0n}G_{-}(t)}-1\right|&=o_p(1). \label{mco2}
	\end{align}
\end{Lem}

\begin{Lem}\label{wtw}
	Suppose Assumptions \ref{ResTrans}$-$\ref{Estimationacc} hold. 
	We have
	\begin{align*}
		W_{j}-\widetilde W_{j}=O_p(c_{np}a_{np}\upsilon_n\sqrt{n\bar{q}_n\log \bar{q}_n}).
	\end{align*}
\end{Lem}

\begin{Lem}\label{klem3}
	Suppose Assumptions \ref{ResTrans}$-$\ref{Estimationacc} hold with $c_{np}a_{np}\upsilon_n\sqrt{n\bar{q}_{n}}(\log{\bar{q}_n})^{3/2+\gamma}\to 0$ for a small $\gamma>0$. Then for any $M>0$, we have
	\begin{align*}
		\sup_{M\leq t\leq G^{-1}(\alpha \eta_n/q_{0n})}\left|\frac{\sum_{j\in \mathcal{A}^c}\mathbb{I}(\widetilde W_{j}\geq t)}{\sum_{j\in \mathcal{A}^c}\mathbb{I}(W_{j}\geq t)}-1\right|=o_p(1),\\
		\sup_{M\leq t\leq G_{-}^{-1}(\alpha \eta_n/q_{0n})}\left|\frac{\sum_{j\in \mathcal{A}^c}\mathbb{I}(\widetilde W_{j}\leq -t)}{\sum_{j\in \mathcal{A}^c}\mathbb{I}(W_{j}\leq -t)}-1\right|=o_p(1).
	\end{align*}
\end{Lem}

%
%
%

\subsection*{Proof of Theorem \ref{exactFDR}}

	We prove the finite-sample FDR control with $L_+$. The result for $L$ can be obtained similarly. The proof technique of this result has been extensively used in \cite{barber2020robust} and  \cite{du2021false}.  Fix $\epsilon >0$ and for any threshold $t>0$, define
\begin{align*}
	R_{\epsilon}(t)=\frac{\sum\nolimits_{j\in\mathcal{A}^c}\mathbb{I}(W_j\geq t,\Delta_j\leq \epsilon)}{1+\sum\nolimits_{j\in\mathcal{A}^c}\mathbb{I}(W_j\leq t)}.
\end{align*}

Define an event that $\mathscr{A} = \left\{\Delta\equiv \max_{j\in\mathcal{A}^c}\Delta_j\leq \epsilon\right\}$. Furthermore, consider a thresholding rule $L=T(\bf{W})$ that maps statistics $\bf{W}$ to a threshold $L\geq 0$. Define
$
L_j = T\left(W_1,\ldots,W_{j-1},\left|W_j\right|,W_{j+1},\ldots,W_p\right)\geq 0, ~~j=1,\ldots, p.
$ Then for the proposed MFSDS with threshold $L_+$, we can write
\begin{align*}
	\frac{\sum\nolimits_{j\in\mathcal{A}^c}\mathbb{I}(W_j\geq L_+,\Delta_j\leq \epsilon)}{1\vee\sum\nolimits_{j}\mathbb{I}(W_j\geq L_+)}
	&=\frac{1+\sum\nolimits_j\mathbb{I}(W_j\leq-L_+)}{1\vee \sum\nolimits_j\mathbb{I}(W_j\geq L_+)}\times\frac{\sum\nolimits_{j\in\mathcal{A}^c}\mathbb{I}(W_j\geq L_+,\Delta_j\leq \epsilon)}{1+\sum\nolimits_j\mathbb{I}(W_j\leq-L_+)}\leq \alpha\times R_{\epsilon}(L_+).
\end{align*}

Next we derive an upper bound for $\mathbb{E}\left\{R_{\epsilon}(L_+)\right\}$. Note that
\begin{align*}
	\begin{split}
		\mathbb{E}\left\{R_{\epsilon}(L_+)\right\}
		&= \mathbb{E}\left\{\frac{\sum\nolimits_{j\in\mathcal{A}^c}\mathbb{I}(W_j\geq L,\Delta_j\leq \epsilon)}{1+\sum\nolimits_{j\in\mathcal{A}^c}\mathbb{I}(W_j\leq-L_+)}\right\}\\
		&=\sum\limits_{j\in\mathcal{A}^c}\mathbb{E}\left\{\frac{\mathbb{I}(W_j\geq L_{j},\Delta_j\leq \epsilon)}{1+\sum\nolimits_{k\in\mathcal{A}^c,k\neq j}\mathbb{I}(W_k\leq-L_{j})}\right\}\\
		&=\sum\limits_{j\in\mathcal{A}^c}\mathbb{E}\left[\mathbb{E}\left\{\frac{\mathbb{I}(W_j\geq L_{j},\Delta_j\leq \epsilon)}{1+\sum\nolimits_{k\in\mathcal{A}^c,k\neq j}\mathbb{I}(W_k\leq-L_{j})}\mid \left|W_j\right|,{\bf{W}}_{-j}\right\}\right]\\
		&=\sum\limits_{j\in\mathcal{A}^c}\mathbb{E}\left\{\frac{{\rm Pr}\left(W_j>0\mid \left|W_j\right|,{\bf{W}}_{-j}\right)\mathbb{I}\left(\left|W_j\right|\geq L_{j},\Delta_j\leq \epsilon\right)}{1+\sum\nolimits_{k\in\mathcal{A}^c,k\neq j}\mathbb{I}(W_k\leq-L_{j})}\right\},
	\end{split}
\end{align*}
where the last step holds since, after conditioning on $\left(\left|W_j\right|,{\bf{W}}_{-j}\right)$, the only unknown quantity is the sign of $W_j$. Recall the definition of  $\Delta_j=\left|\mathbb{P}\left(W_j>0\mid \left|W_j\right|,{\bf{W}}_{-j}\right)-1/2\right|$, we have $\mathbb{P}\left(W_j>0\mid \left|W_j\right|,{\bf{W}}_{-j}\right)\leq 1/2+\Delta_j$. Thus
\begin{align}\label{RL}
	\begin{split}
		\mathbb{E}\left\{R_{\epsilon}(L_+)\right\}
		&\leq \sum\limits_{j\in\mathcal{A}^c}\mathbb{E}\left\{\frac{\left(1/2+\Delta_j\right)\mathbb{I}\left(\left|W_j\right|\geq L_{j},\Delta_j\leq \epsilon\right)}{1+\sum\nolimits_{k\in\mathcal{A}^c,k\neq j}\mathbb{I}(W_k\leq-L_{j})}\right\}\\
		&=\left(\frac{1}{2}+\Delta_j\right)\left[\sum\limits_{j\in\mathcal{A}^c}\mathbb{E}\left\{\frac{\mathbb{I}\left(W_j\geq L_{j},\Delta_j\leq \epsilon\right)}{1+\sum\nolimits_{k\in\mathcal{A}^c,k\neq j}\mathbb{I}(W_k\leq-L_{j})}\right\}+\sum\limits_{j\in\mathcal{A}^c}\mathbb{E}\left\{\frac{\mathbb{I}\left(W_j\leq -L_{j},\Delta_j\leq \epsilon\right)}{1+\sum\nolimits_{k\in\mathcal{A}^c,k\neq j}\mathbb{I}(W_k\leq-L_{j})}\right\}\right]\\
		&\leq\left(\frac{1}{2}+\Delta_j\right)\left[\sum\limits_{j\in\mathcal{A}^c}\mathbb{E}\left\{\frac{\mathbb{I}\left(W_j\geq L_{j},\Delta_j\leq \epsilon\right)}{1+\sum\nolimits_{k\in\mathcal{A}^c,k\neq j}\mathbb{I}(W_k\leq-L_{j})}\right\}+\sum\limits_{j\in\mathcal{A}^c}\mathbb{E}\left\{\frac{\mathbb{I}\left(W_j\leq -L_{j}\right)}{1+\sum\nolimits_{k\in\mathcal{A}^c,k\neq j}\mathbb{I}(W_k\leq-L_{j})}\right\}\right]\\
		&=\left(\frac{1}{2}+\Delta_j\right)\left[\mathbb{E}\left\{R_{\epsilon}(L_+)\right\}+\sum\limits_{j\in\mathcal{A}^c}\mathbb{E}\left\{\frac{\mathbb{I}\left(W_j\leq -L_{j}\right)}{1+\sum\nolimits_{k\in\mathcal{A}^c,k\neq j}\mathbb{I}(W_k\leq-L_{j})}\right\}\right].
	\end{split}
\end{align}
The last expression summation in \eqref{RL} can be simplified. If for all null $j$, $W_j>-L_{j}$, then the summation is zero. Otherwise,
\begin{align*}
	\sum\limits_{j\in\mathcal{A}^c}\mathbb{E}\left\{\frac{\mathbb{I}\left(W_j\leq -L_{j}\right)}{1+\sum\nolimits_{k\in\mathcal{A}^c,k\neq j}\mathbb{I}(W_k\leq-L_{j})}\right\}
	&=\sum\limits_{j\in\mathcal{A}^c}\mathbb{E}\left\{\frac{\mathbb{I}\left(W_j\leq -L_{j}\right)}{1+\sum\nolimits_{k\in\mathcal{A}^c,k\neq j}\mathbb{I}(W_k\leq-L_{k})}\right\}=1,
\end{align*}
where the first equality holds by the fact: for any $j,k$, if $W_j\leq \min\left(L_{j}, L_{k}\right)$ and $W_k\leq -\min\left(L_{j}, L_{k}\right)$, then $L_{+j}=L_{+k}$; see  \cite{barber2020robust} for details.

Accordingly, we have
$\mathbb{E}\left\{R_{\epsilon}(L_+)\right\}\leq \left(1/2+\epsilon\right)\left[\mathbb{E}\left\{R_{\epsilon}(L_+)\right\}+1\right]$. As a result
\begin{align*}
	\mathbb{E}\left\{R_{\epsilon}(L_+)\right\}\leq\frac{1/2+\epsilon}{1/2-\epsilon}\leq 1+4\epsilon.
\end{align*}
Consequently, we can naturally get the conclusions in Theorem \ref{exactFDR}.  $\hfill\square$

\subsection*{Proof of Theorem~\ref{AsyFDR}}

By the definition of FDP, our test is equivalent to select the $j$th variable if $W_j\geq L$, where
\begin{align*}
L=\inf\left\{t\geq 0: \sum_{j}\mathbb{I}(W_j\leq-t)\leq \alpha\left(1\vee \sum_{j}\mathbb{I}(W_j\geq t)\right) \right\}.
\end{align*}

We need to establish an asymptotic bound for this $L$ so that Lemmas \ref{klem1}$-$\ref{klem3} can be applied.
Let $t^*=G^{-1}_{-}(\alpha\eta_{n}/q_{0n})$. By Lemmas \ref{klem2} and \ref{klem3}, it follows that
\begin{align*}
\frac{\alpha\eta_{n}}{q_{0n}}=G_{-}(t^*)=\frac{1}{q_{0n}}\sum_{j\in\mathcal{A}^c}\mathbb{I}(\widetilde{W}_j<-t^*)\{1+o(1)\}=\frac{1}{q_{0n}}\sum_{j\in\mathcal{A}^c}\mathbb{I}(W_j<-t^*)\{1+o(1)\},
\end{align*}
with $c_{np}a_{np}\upsilon_n\bar{q}_{n}\sqrt{n}(\log{\bar{q}_n})^{3/2+\gamma}\to 0$.
On the other hand, for any $j\in\mathcal{C}_{\bf B}$, where $\mathcal{C}_{\bf B}$ is defined in Assumption~\ref{signal} and $\eta_{n}=|\mathcal{C}_{\bf B}|$, we can show that $\text{Pr}(W_j<t^*, j\in\mathcal{C}_{\bf B})\to 0$. In fact, it is straightforward to see that
\begin{align*}
&\text{Pr}\left({W}_j<t^*, \ \mbox{for some}\ j\in\mathcal{C}_{\bf B}\right)\\
&\leq \eta_{n}\text{Pr}\left(\widehat{\bf B}_{1j}^\top \widehat{\bf B}_{2j}/(s_{1\mathcal{S}j}{s}_{2\mathcal{S}j})-{\bf B}_j^\top {\bf B}_j/(s_{1\mathcal{S}j}{s}_{2\mathcal{S}j})<t^*-{\bf B}_j^\top {\bf B}_j/(s_{1\mathcal{S}j}{s}_{2\mathcal{S}j})\right)\\
&\leq \eta_{n}\text{Pr}\left(|{\bf B}_j^\top(\widehat{\bf B}_{1j}-{\bf B}_j+\widehat{\bf B}_{2j}-{\bf B}_j)|+|(\widehat{\bf B}_{1j}-{\bf B}_j)^\top
(\widehat{\bf B}_{2j}-{\bf B}_j)|>{\bf B}^\top_j{\bf B}_j-t^*s_{1\mathcal{S}j}{s}_{2\mathcal{S}j}\right).
\end{align*}

Denote $d_j=\|{\bf B}_j\|_2^2-t^*s_{1j}s_{2j}$. Under Assumption \ref{signal}, it follows that $d_j=\|{\bf B}_j\|_2^2\{1+o(1)\}$. We then get
\begin{align*}
&\text{Pr}\left(|{\bf B}_j^\top(\widehat{\bf B}_{1j}-{\bf B}_j+\widehat{\bf B}_{2j}-{\bf B}_j|)+|(\widehat{\bf B}_{1j}-{\bf B}_j)^\top
(\widehat{\bf B}_{2j}-{\bf B}_j)|>d_j\right)\\
&\leq \text{Pr}\left(|{\bf B}_j^\top(\widehat{\bf B}_{1j}-{\bf B}_j+\widehat{\bf B}_{2j}-{\bf B}_j)|>d_j/2\right)
+\text{Pr}\left(|(\widehat{\bf B}_{1j}-{\bf B}_j)^\top (\widehat{\bf B}_{2j}-{\bf B}_j)|>d_j/2\right)\\
&=:\Lambda_1+\Lambda_2.
\end{align*}
It follows that
\begin{align*}
\Lambda_1&\leq \text{Pr}\left(\|\widehat{\bf B}_{1j}-{\bf B}_j\|_2>d_j/(4\|{\bf B}_j\|_2)\right)
+\text{Pr}\left(\|\widehat{\bf B}_{2j}-{\bf B}_j\|_2>d_j/(4\|{\bf B}_j\|_2)\right)\to 0,\\
\Lambda_2&\leq \text{Pr}\left(\|\widehat{\bf B}_{1j}-{\bf B}_j\|_2>c_{np}\right)
+\Pr\left(\|\widehat{\bf B}_{2j}-{\bf B}_j\|_2> C\sqrt{\bar{q}_n\log{\bar{q}_n}/n}\right)\to 0,
\end{align*}
the above result follows from Assumption \ref{Estimationacc} and the Lemma~\ref{lem2} stated in Supplementary Material.

Thus, we get
$\text{Pr}(\sum_{j}\mathbb{I}(W_j>t^*)\geq \eta_{n})\to 1$. Furthermore, we can conclude that
\begin{align*}
\sum_{j}\mathbb{I}(W_j<-t^*)\lesssim \alpha\eta_{n}\leq \alpha\sum_{j}\mathbb{I}(W_j>t^*).
\end{align*}

Then, we get an upper bound for $L$, that is $L\lesssim t^*$. This implies that the proposed method can detect at least $\sum\nolimits_{j}\mathbb{I}(W_j>t^*)$ signals.

Therefore, by Lemmas \ref{klem1}, \ref{klem2} and \ref{klem3}, we get
\begin{align}\label{fl1}
\frac{\sum_{j\in \mathcal{A}^c}\mathbb{I}(W_j\geq L)}{\sum_{j\in \mathcal{A}^c}\mathbb{I}(W_j\leq-L)}-1 \cp  0.
\end{align}
Write
\begin{align*}
{\rm FDP}&=\frac{\sum_{j\in\mathcal{A}^c}\mathbb{I}\left(W_j\geq L\right)}{1\vee\sum_j\mathbb{I}(W_j\geq L)}=\frac{\sum_{j}\mathbb{I}\left(W_j\leq- L \right)}{1\vee\sum_j\mathbb{I}(W_j\geq L)}\times\frac{\sum_{j\in\mathcal{A}^c}\mathbb{I}\left(W_j\geq L\right)}{\sum_{j}\mathbb{I}\left(W_j\leq-L\right)}
\leq \alpha\times R(L).
\end{align*}

Note that $R(L)\leq {\sum_{j\in\mathcal{A}^c}\mathbb{I}\left(W_j\geq L\right)}/{\sum_{j\in\mathcal{A}^c}\mathbb{I}\left(W_j\leq-L\right)} $, and thus $\mathop{\lim\sup}_{n\to\infty} {\rm FDP}\leq\alpha$ in probability by \eqref{fl1}. Thus the first assertion in Theorem \ref{AsyFDR} is proved.

Further, for any $\epsilon>0$, we have
\begin{align*}
{\rm{FDR}}\leq (1+\epsilon)\alpha R(L)+{\rm{Pr}}\left({\rm{FDP}}\geq (1+\epsilon)\alpha R(L)\right),
\end{align*}
from which the second part of this theorem is proved.
\hfill$\Box$



\bibliographystyle{asa}
\bibliography{SDR_selection-references}

\begin{thebibliography}{38}
\newcommand{\enquote}[1]{``#1''}
\expandafter\ifx\csname natexlab\endcsname\relax\def\natexlab#1{#1}\fi

\bibitem[{Barber and Cand{\`e}s(2015)}]{barber2015controlling}
Barber, R.~F. and Cand{\`e}s, E.~J. (2015), \enquote{Controlling the false
  discovery rate via knockoffs,} \textit{The Annals of Statistics}, 43,
  2055--2085.

\bibitem[{Barber and Cand{\`e}s(2019)}]{barber2019knockoff}
--- (2019), \enquote{A knockoff filter for high-dimensional selective
  inference,} \textit{The Annals of Statistics}, 47, 2504--2537.

\bibitem[{Barber et~al.(2020)Barber, Candes, and Samworth}]{barber2020robust}
Barber, R.~F., Candes, E.~J., and Samworth, R.~J. (2020), \enquote{Robust
  inference with knockoffs,} \textit{The Annals of Statistics}, 48, 1409--1431.

\bibitem[{Benjamini and Hochberg(1995)}]{benjamini1995controlling}
Benjamini, Y. and Hochberg, Y. (1995), \enquote{Controlling the false discovery
  rate: a practical and powerful approach to multiple testing,} \textit{Journal
  of the Royal Statistical Society: Series B (Statistical Methodology)}, 57,
  289--300.

\bibitem[{Cai and Liu(2016)}]{cai2016large}
Cai, T.~T. and Liu, W. (2016), \enquote{Large-scale multiple testing of
  correlations,} \textit{Journal of the American Statistical Association}, 111,
  229--240.

\bibitem[{Candes et~al.(2018)Candes, Fan, Janson, and Lv}]{candes2018panning}
Candes, E., Fan, Y., Janson, L., and Lv, J. (2018), \enquote{Panning for gold:
  ‘model-X’ knockoffs for high dimensional controlled variable selection,}
  \textit{Journal of the Royal Statistical Society: Series B (Statistical
  Methodology)}, 80, 551--577.

\bibitem[{Chen et~al.(2010)Chen, Zou, and Cook}]{chen2010coordinate}
Chen, X., Zou, C., and Cook, R.~D. (2010), \enquote{Coordinate-independent
  sparse sufficient dimension reduction and variable selection,} \textit{The
  Annals of Statistics}, 38, 3696--3723.

\bibitem[{Cook(2004)}]{cook2004testing}
Cook, R.~D. (2004), \enquote{Testing predictor contributions in sufficient
  dimension reduction,} \textit{The Annals of Statistics}, 32, 1062--1092.

\bibitem[{Cook and Ni(2006)}]{cook2006using}
Cook, R.~D. and Ni, L. (2006), \enquote{Using intraslice covariances for
  improved estimation of the central subspace in regression,}
  \textit{Biometrika}, 93, 65--74.

\bibitem[{Cook and Weisberg(1991)}]{cook1991sliced}
Cook, R.~D. and Weisberg, S. (1991), \enquote{Sliced inverse regression for
  dimension reduction: Comment,} \textit{Journal of the American Statistical
  Association}, 86, 328--332.

\bibitem[{Dong(2021)}]{dong2021brief}
Dong, Y. (2021), \enquote{A brief review of linear sufficient dimension
  reduction through optimization,} \textit{Journal of Statistical Planning and
  Inference}, 211, 154--161.

\bibitem[{Du et~al.(2023)Du, Guo, Sun, and Zou}]{du2021false}
Du, L., Guo, X., Sun, W., and Zou, C. (2023), \enquote{False discovery rate
  control under general dependence by symmetrized data aggregation,}
  \textit{Journal of the American Statistical Association}, 118, 607--621.

\bibitem[{Fan et~al.(2012)Fan, Han, and Gu}]{fan2012estimating}
Fan, J., Han, X., and Gu, W. (2012), \enquote{Estimating false discovery
  proportion under arbitrary covariance dependence,} \textit{Journal of the
  American Statistical Association}, 107, 1019--1035.

\bibitem[{Fan and Li(2001)}]{fan2001variable}
Fan, J. and Li, R. (2001), \enquote{Variable selection via nonconcave penalized
  likelihood and its oracle properties,} \textit{Journal of the American
  Statistical Association}, 96, 1348--1360.

\bibitem[{Fan et~al.(2020)Fan, Li, Zhang, and Zou}]{fan2020statistical}
Fan, J., Li, R., Zhang, C.-H., and Zou, H. (2020), \textit{Statistical
  foundations of data science}, CRC press.

\bibitem[{Fan and Lv(2010)}]{fan2010selective}
Fan, J. and Lv, J. (2010), \enquote{A selective overview of variable selection
  in high dimensional feature space,} \textit{Statistica Sinica}, 20, 101--148.

\bibitem[{Fan and Peng(2004)}]{fan2004nonconcave}
Fan, J. and Peng, H. (2004), \enquote{Nonconcave penalized likelihood with a
  diverging number of parameters,} \textit{The Annals of Statistics}, 32,
  928--961.

\bibitem[{Guo et~al.(2024)Guo, Li, Zhang, and Zou}]{guo2024model}
Guo, X., Li, R., Zhang, Z., and Zou, C. (2024), \enquote{Model-Free Statistical
  Inference on High-Dimensional Data,} \textit{Journal of the American
  Statistical Association}, 1--27.

\bibitem[{Javanmard and Montanari(2014)}]{javanmard2014confidence}
Javanmard, A. and Montanari, A. (2014), \enquote{Confidence intervals and
  hypothesis testing for high-dimensional regression,} \textit{The Journal of
  Machine Learning Research}, 15, 2869--2909.

\bibitem[{Khan et~al.(2001)Khan, Wei, Ringner, Saal, Ladanyi, Westermann,
  Berthold, Schwab, Antonescu, Peterson, et~al.}]{khan2001classification}
Khan, J., Wei, J.~S., Ringner, M., Saal, L.~H., Ladanyi, M., Westermann, F.,
  Berthold, F., Schwab, M., Antonescu, C.~R., Peterson, C., et~al. (2001),
  \enquote{Classification and diagnostic prediction of cancers using gene
  expression profiling and artificial neural networks,} \textit{Nature
  Medicine}, 7, 673--679.

\bibitem[{Li and Wang(2007)}]{li2007directional}
Li, B. and Wang, S. (2007), \enquote{On directional regression for dimension
  reduction,} \textit{Journal of the American Statistical Association}, 102,
  997--1008.

\bibitem[{Li(1991)}]{li1991sliced}
Li, K.-C. (1991), \enquote{Sliced inverse regression for dimension reduction,}
  \textit{Journal of the American Statistical Association}, 86, 316--327.

\bibitem[{Li et~al.(2005)Li, Dennis~Cook, and Nachtsheim}]{li2005model}
Li, L., Dennis~Cook, R., and Nachtsheim, C.~J. (2005), \enquote{Model-free
  variable selection,} \textit{Journal of the Royal Statistical Society: Series
  B (Statistical Methodology)}, 67, 285--299.

\bibitem[{Li et~al.(2020)Li, Wen, and Yu}]{li2020selective}
Li, L., Wen, X.~M., and Yu, Z. (2020), \enquote{A selective overview of sparse
  sufficient dimension reduction,} \textit{Statistical Theory and Related
  Fields}, 4, 121--133.

\bibitem[{Luo and Li(2016)}]{luo2016combining}
Luo, W. and Li, B. (2016), \enquote{Combining eigenvalues and variation of
  eigenvectors for order determination,} \textit{Biometrika}, 103, 875--887.

\bibitem[{Meinshausen et~al.(2009)Meinshausen, Meier, and
  B{\"u}hlmann}]{meinshausen2009p}
Meinshausen, N., Meier, L., and B{\"u}hlmann, P. (2009), \enquote{P-values for
  high-dimensional regression,} \textit{Journal of the American Statistical
  Association}, 104, 1671--1681.

\bibitem[{Shao et~al.(2007)Shao, Cook, and Weisberg}]{shao2007marginal}
Shao, Y., Cook, R.~D., and Weisberg, S. (2007), \enquote{Marginal tests with
  sliced average variance estimation,} \textit{Biometrika}, 94, 285--296.

\bibitem[{Svetulevi{\v{c}}ien{\.e}(1982)}]{svetulevivciene1982multidimensional}
Svetulevi{\v{c}}ien{\.e}, V. (1982), \enquote{Multidimensional local limit
  theorems for probabilities of moderate deviations,} \textit{Lithuanian
  Mathematical Journal}, 22, 416--420.

\bibitem[{Tibshirani(1996)}]{tibshirani1996regression}
Tibshirani, R. (1996), \enquote{Regression shrinkage and selection via the
  lasso,} \textit{Journal of the Royal Statistical Society: Series B
  (Statistical Methodology)}, 58, 267--288.

\bibitem[{Wasserman and Roeder(2009)}]{wasserman2009high}
Wasserman, L. and Roeder, K. (2009), \enquote{High dimensional variable
  selection,} \textit{The Annals of Statistics}, 37, 2178--2201.

\bibitem[{Wu and Li(2011)}]{wu2011asymptotic}
Wu, Y. and Li, L. (2011), \enquote{Asymptotic properties of sufficient
  dimension reduction with a diverging number of predictors,}
  \textit{Statistica Sinica}, 2011, 707--703.

\bibitem[{Xia et~al.(2002)Xia, Tong, Li, and Zhu}]{xia2002adaptive}
Xia, Y., Tong, H., Li, W., and Zhu, L.-X. (2002), \enquote{An adaptive
  estimation of dimension reduction space,} \textit{Journal of the Royal
  Statistical Society: Series B (Statistical Methodology)}, {64}, 363--410.

\bibitem[{Yin and Cook(2002)}]{yin2002dimension}
Yin, X. and Cook, R.~D. (2002), \enquote{Dimension reduction for the
  conditional kth moment in regression,} \textit{Journal of the Royal
  Statistical Society: Series B (Statistical Methodology)}, 64, 159--175.

\bibitem[{Yu and Dong(2016)}]{yu2016model}
Yu, Z. and Dong, Y. (2016), \enquote{Model-free coordinate test and variable
  selection via directional regression,} \textit{Statistica Sinica}, 26,
  1159--1174.

\bibitem[{Yu et~al.(2016{\natexlab{a}})Yu, Dong, and Shao}]{yu2016marginal}
Yu, Z., Dong, Y., and Shao, J. (2016{\natexlab{a}}), \enquote{On marginal
  sliced inverse regression for ultrahigh dimensional model-free feature
  selection,} \textit{The Annals of Statistics}, 44, 2594--2623.

\bibitem[{Yu et~al.(2016{\natexlab{b}})Yu, Dong, and Zhu}]{yu2016trace}
Yu, Z., Dong, Y., and Zhu, L.-X. (2016{\natexlab{b}}), \enquote{Trace pursuit:
  A general framework for model-free variable selection,} \textit{Journal of
  the American Statistical Association}, 111, 813--821.

\bibitem[{Zhu(2020)}]{zhu2020review}
Zhu, L. (2020), \enquote{Review of sparse sufficient dimension reduction:
  comment,} \textit{Statistical Theory and Related Fields}, 4, 134--134.

\bibitem[{Zou(2006)}]{zou2006adaptive}
Zou, H. (2006), \enquote{The adaptive lasso and its oracle properties,}
  \textit{Journal of the American Statistical Association}, 101, 1418--1429.

\end{thebibliography}

\newpage
\setcounter{page}{1}
\beginsupplement


\centerline{\bf Supplementary Material for ``Model-free variable selection in }
\centerline{\bf sufficient dimension reduction via FDR control"}

\vspace{1cm}
This Supplementary Material contains the proofs of some technical lemmas and additional simulation results.

\renewcommand{\thesection}{}
\setcounter{Lem}{0}
\renewcommand{\theLem}{S.\arabic{Lem}}

\renewcommand{\theequation}{}
\setcounter{equation}{0}
\renewcommand{\theequation}{S.\arabic{equation}}

\renewcommand{\thesection}{}
\setcounter{figure}{0}
\renewcommand{\thefigure}{S\arabic{figure}}

\subsection*{S1. Additional lemmas}

The first lemma is the standard Bernstein's inequality.
\begin{Lem}[Bernstein's inequality]\label{lemma:bernstein}
Let $X_1,\ldots,X_n$ be independent centered random variables a.s. bounded by $A<\infty$ in absolute value. Let $\sigma^2=n^{-1}\sum_{i=1}^n\mathbb{E}(X_i^2)$. Then for all $x>0$,
\begin{align*}
  {{\Pr}}\left(\sum_{i=1}^n X_i\geq x\right)\leq\exp\left(-\frac{x^2}{2n\sigma^2+2Ax/3}\right).
\end{align*}
\end{Lem}


The second one is a moderate deviation result for the mean of random vector;
See Theorem 1 in \citet{svetulevivciene1982multidimensional}.

\begin{Lem}[Moderate deviation for the independent sum]\label{mdm}
Suppose that $\X_1,\ldots,\X_n\in\mathbb{R}^s$ are independent identically distributed random vectors with mean zero and identity covariance matrix. Let $\mathbb{E}(\|\X\|_2^q)<\infty$ for some $q>2+c^2$ with $c^2>0$. Then in the domain $\|{\bf x}\|_2\leq c\sqrt{\log n}$, one has uniformly
\begin{align*}
\frac{p_n({\bf x})}{(2\pi)^{-s/2}\exp(-\|{\bf x}\|^2_2/2)}\to 1,
\end{align*}
as $n\to\infty$. Here $p_n({\bf x})$ is the density function of $\sum_{i=1}^n \X_i/\sqrt n$.
\end{Lem}

\subsection*{S2. Useful lemmas}

For simplicity of notation, the constant $c$ and $C$ may be slightly abused hereafter. The following lemma establishes uniform bounds for $\widetilde{\bf B}_{2j}$, $j=1,\ldots,p$.

\begin{Lem}\label{lem1}
	Suppose Assumptions \ref{ResTrans}$-$\ref{DesignMatrix} hold. Then, for a large $C>0$
	\begin{align*}
		{\rm{Pr}}\left(\widetilde s_{2\mathcal{S}j}^{-1}\|\widetilde{\bf B}_{2j}\|_2>\sigma\sqrt{C\log \bar{q}_n}\mid \mathcal{D}_1\right)=o(1/\bar{q}_n),
	\end{align*}
	holds uniformly in $\mathcal{S}\bigcap\mathcal{A}^c$. Here denote $\sigma^2=H\max_{h=1}^H\bm{e}_j^\top\bm\Sigma_{\mathcal{S}}^{-1}\mathbb{E}[\X_{\mathcal{S}}\X_{\mathcal{S}}^\top\varepsilon_h^2]\bm\Sigma_{\mathcal{S}}^{-1}\bm{e}_j/(\bm{e}^\top_j\bm\Sigma_{\mathcal{S}}^{-1}\bm{e}_j)$.
\end{Lem}

\begin{proof}
	Note that
	\begin{align*}
		{\rm{Pr}}\left(\widetilde s_{2\mathcal{S}j}^{-1}\|\widetilde{\bf B}_{2j}\|_2> x\right)
		&={\rm{Pr}}\left(\sum_{h=1}^H\widetilde{\bf B}_{2j,h}^2>x^2\widetilde s_{2\mathcal{S}j}^2\right)\\
		&\leq
		{\rm{Pr}}\left(\widetilde{\bf B}_{2j,h}^2>x^2\widetilde s_{2\mathcal{S}j}^2/H,\,\,\,\exists\,\,h=1,\ldots,H\right)\\
		&\leq H\max_h{\rm{Pr}}\left(|\widetilde{\bf B}_{2j,h}|>x\widetilde s_{2j}/\sqrt H\right).
	\end{align*}
	Recall that
	$\widetilde{\bf B}_{2j,h}=\bm{e}_j^\top\bm\Sigma_{\mathcal{S}}^{-1}{n}^{-1}\sum_{i=1}^n \X_{2\mathcal{S}i}\varepsilon_{ih}$.
	Let $\mathcal{B}=\{\max_{1\leq h\leq H}\max_{1\leq i\leq n}\|\bm\Sigma_{\mathcal{S}}^{-1}\X_{2\mathcal{S}i}\varepsilon_{ih}\|_{\infty}\leq m_n\}$, where $m_n=(n\bar{q}_n)^{1/\varpi+\gamma}K_{n}$ for some small $\gamma>0$. In what follows, we first work on the case of the occurrence of $\mathcal{B}$. Denote $\sigma^2_h=\bm{e}_j^\top\bm\Sigma_{\mathcal{S}}^{-1}\mathbb{E}[\X_{\mathcal{S}}\X_{\mathcal{S}}^\top\varepsilon_h^2]\bm\Sigma_{\mathcal{S}}^{-1}\bm{e}_j$ and $\sigma^2=\max_h\sigma^2_h H/(\bm{e}^\top_j\bm\Sigma_{\mathcal{S}}^{-1}\bm{e}_j)$.
	The Bernstein's inequality in Lemma \ref{lemma:bernstein} yields that
	\begin{align*}
		{\rm Pr}\left(|\widetilde{\bf B}_{2j,h}|>x\widetilde s_{2\mathcal{S}j}/\sqrt H\ \ \mbox{for some}\ j\mid \mathcal{D}_1\right)
		&\leq q_n\max_{j\in\mathcal{S}}{\rm Pr}\left(\left|\sum_{i=1}^n \bm{e}_j^\top\bm\Sigma_{\mathcal{S}}^{-1}\X_{2\mathcal{S}i}\varepsilon_{ih}\right|>nx\tilde s_{2\mathcal{S}j}/\sqrt H\mid \mathcal{D}_1\right)\\
		&\leq 2\bar{q}_n\max_{j\in\mathcal{S}}\exp\left\{-\frac{n^2x^2\tilde s^2_{2\mathcal{S}j}/H}{2n\sigma^2_h+2nx\tilde s_{2\mathcal{S}j}/\sqrt H\cdot m_n/3}\right\}\\
		&\leq 2\bar{q}_n\max_{j\in\mathcal{S}}\exp\left\{-\frac{x^2\bm{e}^\top_j\bm\Sigma_{\mathcal{S}}^{-1}\bm{e}_j/H}{2\sigma^2_h+2x\sqrt{\bm{e}^\top_j\bm\Sigma_{\mathcal{S}}^{-1}\bm{e}_j/nH}\cdot m_n/3}\right\}\\
		&=2\bar{q}_n\max_{j\in\mathcal{S}}\exp\left\{-\frac{x^2}{2\sigma^2+2m_nx/(3\sqrt n)\sqrt{H/\bm{e}_j^\top\bm\Sigma_{\mathcal{S}}^{-1}\bm{e}_j}}\right\}\\
		&=o(1/\bar{q}_n),
	\end{align*}
	holds uniformly in $\mathcal{S}$. here we use the condition $m_n\sqrt{\log \bar{q}_n/n}=o(1)$ which is implied by Assumption \ref{moment} and $x=\sqrt{C\log \bar{q}_n}$. Next, we turn to consider the case on $\mathcal{B}^c$.
	By Assumption \ref{moment} and Markov's inequality
	\begin{align*}
		{\rm Pr}(\mathcal{B}^c)\leq nH\max_{i,h}{\rm Pr}(\|\bm\Sigma^{-1}\X_{2\mathcal{S}i}\varepsilon_{ih}\|_{\infty}^{\varpi}>m_n^{\varpi})=o(1/\bar{q}_n).
	\end{align*}
	
	The lemma is proved. \hfill$\Box$
\end{proof}

The next lemma establishes uniform bounds for $\widehat{\bf B}_{2j}$, $j=1,\ldots, p$.

\begin{Lem}\label{lem2}
Suppose Assumptions \ref{ResTrans}$-$\ref{DesignMatrix} hold. Then, for a large $C>0$
\begin{align*}
	{\rm Pr}\left(s_{2\mathcal{S}j}^{-1}\|\widehat{\bf B}_{2j}-{\bf B}_{j}\|_2> \sigma\sqrt{C\bar{q}_n\log \bar{q}_n}\mid \mathcal{D}_1\right)=o(1/{\bar{q}_n}),
\end{align*}
uniformly holds for $j\in \mathcal{S}$.
\end{Lem}

\begin{proof}
By the similar proof of Lemma \ref{lem1} and Assumption \ref{DesignMatrix}, we have
\begin{align*}
	{\rm{Pr}}\left(s_{2\mathcal{S}j}^{-1}\|\widehat{\bf B}_{2j}-{\bf B}_{j}\|_2> x\right)
	&\leq H\max_h{\rm{Pr}}\left(|\widehat{\bf B}_{2j,h}-{\bf B}_{j,h}|>xs_{2\mathcal{S}j}/\sqrt H\right)\\
	&\leq H\max_h{\rm{Pr}}\left(|\widehat{\bf B}_{2j,h}-{\bf B}_{j,h}|>x/\sqrt {n\bar{\kappa}H}\right).
\end{align*}

Conditionally on $\mathcal{D}_1$, we know that $
\widehat{\bf B}_{2j,h}-{\bf B}_{j,h}=\bm{e}_j^\top\left({n}^{-1}\sum_{i=1}^n \X_{2\mathcal{S}i}\X_{2\mathcal{S}i}^\top\right)^{-1}{n}^{-1}\sum_{i=1}^n \X_{2\mathcal{S}i}\varepsilon_{ih}.$ Using Assumption~\ref{DesignMatrix}
\begin{align*}
	\max_{j\in\mathcal{S}}|\widehat{\bf B}_{2j,h}-{\bf B}_{j,h}|\leq \underline{\kappa}^{-1}\left|\frac{1}{n}\sum_{i=1}^n \bm{e}_j^\top\X_{2\mathcal{S}i}\varepsilon_{ih}\right|
	\leq \underline{\kappa}^{-1}\sqrt{\bar{q}_n}\max_{j\in\mathcal{S}}|\zeta_j|,
\end{align*}
where $\bm\zeta=n^{-1}\sum_{i=1}^n \bm{e}_j^\top\X_{2\mathcal{S}i}\varepsilon_{ih}$.	Let $z=c\sqrt{\log{\bar{q}_n}/n}$ with $c$ being very large and $L$ is a positive constant.
By Lemma \ref{lemma:bernstein} again, we have
\begin{align*}
	{\rm{Pr}}\left(\max_{j\in\mathcal{S}}|\zeta_j|>z\right)
	&\leq \bar{q}_n\max_{j\in\mathcal{S}}{\rm{Pr}}\left(\left|\sum\limits_{i=1}^{n}\bm{e}_j^\top\X_{2\mathcal{S}i}\varepsilon_{ih}\right|>nz\right)\\
	&\leq2{\bar{q}_n}\max_{j\in\mathcal{S}}\exp\left\{-\frac{n^2z^2}{2n\mathbb{E}[(\bm{e}_j^\top\X_{2\mathcal{S}i})^2\varepsilon_{ih}^2]+2(\max\nolimits_{1\leq h\leq H,1\leq i\leq n}|\bm{e}_j^\top\X_{2\mathcal{S}i}\varepsilon_{ih}|)\times nz/3}\right\}\to 0,
\end{align*}
where we use the condition $\max\nolimits_{1\leq h\leq H, 1\leq i\leq n}\E\|\X_{2\mathcal{S}i}\varepsilon_{ih}\|_{\infty}^{\varpi}\leq LK_n^{\varpi}$ and $(n\bar{q}_n)^{1/\varpi+\gamma}\times\sqrt{\log \bar{q}_n/n}\rightarrow 0$ implied by Assumption~\ref{moment}. Thus, for some large $C>0$, we have
\begin{align*}
	{\rm{Pr}}\left(s_{2\mathcal{S}j}^{-1}\|\widehat{\bf B}_{2j}-{\bf B}_{j}\|_2> x\right)
	&\leq H\max_h{\rm{Pr}}\left(|\widetilde{\bf B}_{2j,h}-{\bf B}_{j,h}|>x/(\bar{\kappa}\sqrt {nH})\right) \\
	&\leq H\max_h{\rm{Pr}}\left(\max_{j\in\mathcal{S}}\left|\frac{1}{n}\sum_{i=1}^n \bm{e}_j^\top\X_{2\mathcal{S}i}\varepsilon_{ih}\right|>x\underline{\kappa}/(\bar{\kappa}\sqrt {n\bar{q}_nH})\right)\\
	&= H\max_h{\rm{Pr}}\left(\max_{j\in\mathcal{S}}|\zeta_j|>z\right)\to 0,
\end{align*}
hold if $x=C\sqrt{\bar{q}_n\log{\bar{q}_n}}$. \hfill$\Box$
\end{proof}

\subsection*{S3. Proofs of lemmas in Appendix}

\bigskip
\noindent
\textbf{Proof of Lemma~\ref{klem1}}.
\noindent

Define $b_n=\sigma\sqrt{C\log \bar{q}_n}$ where $C>0$. Hereafter for simplicity, we denote $T_{1j}=\widehat{\bf B}_{1j}/s_{1\mathcal{S}j}$ and
$T_{2j}=\widetilde{\bf B}_{2j}/\widetilde s_{2\mathcal{S}j}$. Thus $\widetilde W_j=T_{1j}^\top T_{2j}$. We observe that
\begin{align*}
	\frac{G(t)}{G_{-}(t)}-1=&\frac{\sum\nolimits_{j\in\mathcal{A}^c}\left\{\text{Pr}(\widetilde{W}_j\geq t\mid \mathcal{D}_1)-\text{Pr}(\widetilde{W}_j\leq-t\mid \mathcal{D}_1)\right\}}{\sum\nolimits_{j\in\mathcal{A}^c}\text{Pr}(\widetilde{W}_j\leq-t\mid \mathcal{D}_1)}\\
	=&\frac{\sum_{j\in\mathcal{A}^c}\left\{\text{Pr}(T_{1j}^\top T_{2j}\geq t,\|T_{2j}\|_2\leq b_n\mid\mathcal{D}_1)-\text{Pr}(T_{1j}^\top T_{2j}\leq -t,\|T_{2j}\|_2\leq b_n\mid \mathcal{D}_1)\right\}}{q_{0n}G_{-}({t})}\\
	&+\frac{\sum_{j\in\mathcal{A}^c}\left\{\text{Pr}(T_{1j}^\top T_{2j}\geq t,\|T_{2j}\|_2> b_n\mid \mathcal{D}_1)-\text{Pr}(T_{1j}^\top T_{2j}\leq -t,\|T_{2j}\|_2> b_n\mid \mathcal{D}_1)\right\}}{q_{0n}G_{-}({t})}\\
	:=& \Lambda_1+\Lambda_2.
\end{align*}

Note that $G_{-}(t)$ is a decreasing function by definition. Firstly, for the term $\Lambda_2$, by Lemma \ref{lem1}, we have
\begin{align*}
	\frac{\sum_{j\in\mathcal{A}^c}\text{Pr}(T_{1j}^\top T_{2j}\geq t,\|T_{2j}\|_2> b_n \mid \mathcal{D}_{1})}{q_{0n}G_{-}({t})}\leq \frac{\sum_{j\in\mathcal{A}^c}\text{Pr}(\|T_{2j}\|_2>b_n\mid \mathcal{D}_1)}{\alpha\eta_{n}}
	\lesssim \frac{\bar{q}_{n}\times o(1/\bar{q}_n)}{\eta_{n}},
\end{align*}
which implies $\Lambda_2=o(1)$. Recall that $
T_{2j}^\top={\bm{e}_j^\top\bm\Sigma_{\mathcal{S}}^{-1}\sum_{i=1}^n \X_{2\mathcal{S}i}\bm\varepsilon_{2i}^\top}/{\sqrt{n\bm{e}_j\bm\Sigma_{\mathcal{S}}^{-1}\bm{e}_j}}$. By Lemma \ref{mdm} and Assumption \ref{moment}, we obtain that
\begin{align*}
	\frac{\text{Pr}(T_{1j}^\top T_{2j}\geq t, \|T_{2j}\|_2\leq b_n\mid \mathcal{D}_1)}{\text{Pr}(T_{1j}^\top U_j\geq t, \|U_j\|_2\leq b_n\mid \mathcal{D}_1)}\to 1,
\end{align*}
where $U_j\sim N_H(0,\widetilde{\bm\Sigma}_j)$.
Here $\widetilde\sigma_{lh}=\bm{e}_j^\top\bm\Sigma_{\mathcal{S}}^{-1}\mathbb{E}[\bm\varepsilon_l\X_{\mathcal{S}}^\top\X_{\mathcal{S}}\bm\varepsilon_h^\top]\bm\Sigma_{\mathcal{S}}^{-1}\bm{e}_j/
\bm{e}_j^\top\bm\Sigma_{\mathcal{S}}^{-1}\bm{e}_j$ and $\widetilde{\bm\Sigma}_j=(\widetilde\sigma_{lh})_{l,h=1}^H$.	Similarly we get
\begin{align*}
	\frac{\text{Pr}(T_{1j}^\top T_{2j}\leq -t, \|T_{2j}\|_2\leq b_n\mid \mathcal{D}_1)}{\text{Pr}(T_{1j}^\top U_j\leq -t, \|U_j\|_2\leq b_n\mid \mathcal{D}_1)}\to 1.
\end{align*}

Note that ${\Pr}(T_{1j}^\top U_j\leq -t, \|U_j\|_2\leq b_n\mid \mathcal{D}_1)=\Pr(T_{1j}^\top U_j\geq t, \|U_j\|_2\leq b_n\mid \mathcal{D}_1)$, from which we get $\Lambda_1=o(1)$ and accordingly we can claim the assertion. \hfill$\Box$

\bigskip
\noindent
\textbf{Proof of Lemma \ref{klem2}}.
\noindent

We prove the first formula and the second one can be deduced similarly. By the proof of Lemma~\ref{klem1}, it suffices to show that
\begin{align*}
	G(t)=\frac{1}{q_{0n}}\sum\limits_{j\in\mathcal{A}^c}{\text{Pr}}(\widetilde{W}_j\geq t, \|T_{2j}\|_2\leq b_n\mid \mathcal{D}_1)\left\{1+o(1)\right\}:=\bar{G}(t)\left\{1+o(1)\right\}.
\end{align*}
Accordingly
\begin{align*}
	\frac{1}{q_{0n}}\sum\limits_{j\in\mathcal{A}^c}\mathbb{I}(\widetilde{W}_j\geq t)=\frac{1}{q_{0n}}\sum\limits_{j\in\mathcal{A}^c}\mathbb{I}(\widetilde{W}_j\geq t,\|T_{2j}\|_2\leq b_n)\left\{1+o_p(1)\right\}.
\end{align*}

Thus, we need to show the following assertion
\begin{align*}
	\sup_{0\leq t\leq G^{-1}(\alpha \eta_n/q_{0n})}\left|\frac{\sum_{j\in \mathcal{A}^c}\mathbb{I}(\widetilde W_{j}\geq t,\|T_{2j}\|_2\leq b_n)}{q_{0n}\bar{G}(t)}-1\right|&=o_p(1).
\end{align*}

Note that the $\bar{G}(t)$ is a decreasing and continuous function. Let $a_p=\alpha\eta_{n}$,$z_0<z_1<\cdots<z_{h_n}\leq 1$ and $t_i=\bar{G}^{-1}(z_i)$, where $z_0=a_p/q_{0n}, z_i=a_p/q_{0n}+b_p\exp(i^{\zeta})/q_{0n}, h_n=\{\log((q_{0n}-a_p)/b_p)\}^{1/\zeta}$ with $b_p/a_p\to 0$ and $0<\zeta<1$. Note that $\bar{G}(t_i)/\bar{G}(t_{i+1})=1+o(1)$ uniformly in $i$. It is therefore enough to derive the convergence rate of the following formula
\begin{align*}
	D_n=\sup_{0\leq i\leq h_n}\left|\frac{\sum_{j\in \mathcal{A}^c}\left\{\mathbb{I}(\widetilde W_{j}\geq t_i,\|T_{2j}\|_2\leq b_n)-\text{Pr}(\widetilde W_{j}\geq t_i,\|T_{2j}\|_2\leq b_n\mid\mathcal{D}_1)\right\}}{q_{0n}\bar{G}(t_i)}\right|.
\end{align*}

Define $\mathcal{B}=\left\{\|T_{2j}\|_2\leq b_n, j\in\mathcal{A}^c\right\}$ and then we have
\begin{align*}
	D(t)&=\mathbb{E}\left[\left(\sum_{j\in \mathcal{A}^c}\left\{\mathbb{I}(\widetilde W_{j}>t,\|T_{2j}\|_2\leq b_n)-{\Pr}(\widetilde W_{j}>t,\|T_{2j}\|_2\leq b_n\mid \mathcal{D}_1)\right\}\right)^2\mid \mathcal{D}_1\right]\\
	&=\sum_{j\in \mathcal{A}^c}\sum_{l\in \mathcal{A}^c}\left\{{\Pr}(\widetilde W_{j}>t, \widetilde W_{l}>t\mid\mathcal{D}_1,\mathcal{B})-\Pr(\widetilde W_{j}>t\mid\mathcal{D}_1,\mathcal{B})\Pr(\widetilde W_{l}>t\mid\mathcal{D}_1,\mathcal{B})\right\}\left\{1+o(1)\right\}.
\end{align*}

Further define $\mathcal{M}_{j}=\{l\in\mathcal{A}^c:|\rho_{jl}|\geq C(\log n)^{-2-\nu}\}$. Here $\rho_{jl}$ denotes the conditional correlation between $\widetilde W_{j}$ and $\widetilde W_{l}$ given $\mathcal{D}_1$.
Recall that$
\widetilde W_{j}={\bm{e}_j^\top\bm\Sigma_{\mathcal{S}}^{-1}\sum_{i=1}^n \X_{2\mathcal{S}i}\bm\varepsilon_{2i}^\top T_{1j}}/{\sqrt{n\bm{e}_j^\top\bm\Sigma_{\mathcal{S}}^{-1}\bm{e}_j}}.$	It follows that
\begin{align*}
	\text{var}(\widetilde W_{j}\mid \mathcal{D}_1)&=\frac{\bm{e}_j^\top\bm\Sigma_{\mathcal{S}}^{-1}\mathbb{E}[\X_{\mathcal{S}}\X_{\mathcal{S}}^\top(\bm\varepsilon^\top T_{1j})^2]\bm\Sigma_{\mathcal{S}}^{-1}\bm{e}_j}{\bm{e}_j^\top\bm\Sigma_{\mathcal{S}}^{-1}\bm{e}_j},~~
	\text{cov}(\widetilde W_{j},\widetilde W_{l}\mid \mathcal{D}_1)=\frac{\bm{e}_j^\top\bm\Sigma_{\mathcal{S}}^{-1}\mathbb{E}[\X_{\mathcal{S}}\X_{\mathcal{S}}^\top\bm\varepsilon^\top T_{1j}T_{1l}^\top\bm\varepsilon]\bm\Sigma_{\mathcal{S}}^{-1}\bm{e}_l}{\sqrt{\bm{e}_j^\top\bm\Sigma_{\mathcal{S}}^{-1}\bm{e}_j}\cdot\sqrt{\bm{e}_l^\top\bm\Sigma_{\mathcal{S}}^{-1}\bm{e}_l}}.
\end{align*}

So we get
\begin{align*}
	\rho_{jl}=\frac{\bm{e}_j^\top\bm\Sigma_{\mathcal{S}}^{-1}\mathbb{E}[\X_{\mathcal{S}}\X_{\mathcal{S}}^\top\bm\varepsilon^\top T_{1j}T_{1l}^\top\bm\varepsilon]\bm\Sigma_{\mathcal{S}}^{-1}\bm{e}_l}{
		\sqrt{\bm{e}_j^\top\bm\Sigma_{\mathcal{S}}^{-1}\mathbb{E}[\X_{\mathcal{S}}\X_{\mathcal{S}}^\top(\bm\varepsilon^\top T_{1j})^2]\bm\Sigma_{\mathcal{S}}^{-1}\bm{e}_j}\cdot\sqrt{\bm{e}_l^\top\bm\Sigma_{\mathcal{S}}^{-1}\mathbb{E}[\X_{\mathcal{S}}\X_{\mathcal{S}}^\top(\bm\varepsilon^\top T_{1l})^2]\bm\Sigma_{\mathcal{S}}^{-1}\bm{e}_l}}.
\end{align*}
By Assumption \ref{dependence}
\begin{align*}
	D(t)
	&\leq\sum\limits_{j\in\mathcal{A}^c}\sum\limits_{j\in\mathcal{M}}\text{Pr}(\widetilde W_{j}>t\mid\mathcal{D}_1) +\sum_{j\in \mathcal{A}^c}\sum_{l\in \mathcal{M}_{j}^c}\left\{\text{Pr}(\widetilde W_{j}>t, \widetilde W_{l}>t\mid\mathcal{D}_1)-\text{Pr}(\widetilde W_{j}>t\mid\mathcal{D}_1)\text{Pr}(\widetilde W_{l}>t\mid\mathcal{D}_1)\right\}\\
	&\leq r_p q_{0n} G(t)+\sum_{j\in \mathcal{A}^c}\sum_{l\in \mathcal{M}_{j}^c}\left\{\text{Pr}(\widetilde W_{j}>t, \widetilde W_{l}>t\mid\mathcal{D}_1,\mathcal{B})-\text{Pr}(\widetilde W_{j}>t\mid\mathcal{D}_1,\mathcal{B})\text{Pr}(\widetilde W_{l}>t\mid\mathcal{D}_1,\mathcal{B})\right\}.
\end{align*}

While for each $j\in \mathcal{A}^c$ and $l\in \mathcal{M}_{j}^c$, conditional on $\mathcal{D}_1$, by Lemma 1 in \citet{cai2016large} we have
\begin{align*}
	&\left|\frac{\text{Pr}(\widetilde W_{j}>t, \widetilde W_{l}>t\mid\mathcal{D}_1,\mathcal{B})-\text{Pr}(\widetilde W_{j}>t\mid\mathcal{D}_1,\mathcal{B})\text{Pr}(\widetilde W_{l}>t\mid\mathcal{D}_1,\mathcal{B})}{\text{Pr}(\widetilde W_{j}>t\mid\mathcal{D}_1,\mathcal{B})\text{Pr}(\widetilde W_{l}>t\mid\mathcal{D}_1,\mathcal{B})}\right|\\
	&=\left|\frac{\text{Pr}(\widetilde W_{j}>t, \widetilde W_{l}>t\mid\mathcal{D}_1,\mathcal{B})-\text{Pr}(\widetilde W_{j}>t\mid\mathcal{D}_1,\mathcal{B})\text{Pr}(\widetilde W_{l}>t\mid\mathcal{D}_1,\mathcal{B})}{\bar{G}^2(t)}\right|
	\leq A_{n},
\end{align*}
uniformly holds, where $A_{n}=(\log n)^{-1-\nu_1}$ for $\nu_1=\min(\nu,1/2)$.	From the above results and Chebyshev's inequality, for $\xi>0$ we have
\begin{align*}
	\text{Pr}(D_n\geq \xi\mid\mathcal{D}_{1})&\leq \sum_{i=0}^{h_n}
	\text{Pr}\left(\left|\frac{\sum_{j\in \mathcal{A}^c} \left\{\mathbb{I}(W_j>t_i,\|T_{2j}\|_2\leq b_n)-\text{Pr}(W_j>t_i,\|T_{2j}\|_2\leq b_n\mid\mathcal{D}_{1})\right\}}{q_{0n} \bar{G}(t_i)}\right|\geq \xi\mid\mathcal{D}_{1}\right)\\
	&\leq  \frac{1}{\xi^2}\sum_{i=0}^{h_n}\frac{1}{q_{0n}^2 \bar{G}^2(t_i)}D(t_i)\\
	&\leq  \frac{1}{\xi^2}\sum_{i=0}^{h_n}\frac{1}{q_{0n}^2 \bar{G}^2(t_i)}\left\{r_pq_{0n}\bar{G}(t_i)+q_{0n}^2\bar{G}^2(t_i)A_n\right\}\\
	&\leq \frac{1}{\xi^2}\left\{\sum_{i=0}^{h_n}\frac{r_p}{q_{0n} \bar{G}(t_i)}+h_nA_n\right\}.
\end{align*}

Moreover, observe that
\begin{align*}
	\sum_{i=0}^{h_n}\frac{1}{q_{0n}\bar{G}(t_i)}=\frac{1}{a_p}
	+\sum_{i=1}^{h_n}\frac{1}{a_p+b_pe^{i^{\zeta}}}\lesssim b_p^{-1}.
\end{align*}

Note that $\zeta$ can be arbitrarily close to 1 such that $h_n A_n\to 0$.  Because $b_p$ can be made arbitrarily large as long as $b_p/a_p\to 0$, we have $D_n=o_p(1)$ providing that $r_p/b_{p}\to 0$. We need to point out that the results hold uniformly in $\mathcal{D}_1$ and thus are actually unconditional results.
$\hfill\Box$

\bigskip
\noindent
\textbf{Proof of Lemma \ref{wtw}}.
\noindent

By definition
\begin{align*}
W_j-\widetilde{W}_j= \frac{\widehat{\bf B}_{1j}^\top\widehat{\bf B}_{2j}}{ s_{1\mathcal{S}j}s_{2\mathcal{S}j}}-\frac{\widehat{\bf B}_{1j}^\top\widetilde{\bf B}_{2j}}{ s_{1\mathcal{S}j}\widetilde{s}_{2\mathcal{S}j}}=\frac{\widehat{\bf B}_{1j}^\top}{s_{1\mathcal{S}j}}\left(\frac{\widehat{\bf B}_{2j}}{s_{2\mathcal{S}j}}-\frac{\widetilde{\bf B}_{2j}}{\widetilde s_{2\mathcal{S}j}}\right),
\end{align*}
where 
$\|{\widehat{\bf B}_{1j}}/{s_{1\mathcal{S}j}}\|_\infty=O_p(\sqrt{n}c_{np})$ uniformly holds
for $j\in\mathcal{S}$ by Assumption \ref{Estimationacc}. 

Recall that
\begin{align*}
\frac{\widehat{\bf B}_{2j}^\top}{s_{2\mathcal{S}j}}-\frac{\widetilde{\bf B}_{2j}^\top}{\widetilde s_{2\mathcal{S}j}}
&=\frac{\bm{e}^\top_j{\bf A}^{-1}\sum\limits_{i=1}^n \X_{2\mathcal{S}i}\bm\varepsilon_i^\top}{\sqrt{n\bm{e}^\top_j{\bf A}^{-1}\bm{e}_j}}-\frac{\bm{e}^\top_j\bm\Sigma_{\mathcal{S}}^{-1}\sum\limits_{i=1}^n \X_{2\mathcal{S}i}\bm\varepsilon_i^\top}{\sqrt{n\bm{e}^\top_j\bm\Sigma_{\mathcal{S}}^{-1}\bm{e}_j}}
=\frac{1}{\sqrt{n}}\bm{e}^\top_j({\bf A}^{-1}-{\bm \Sigma}_{\mathcal{S}}^{-1})\sum\limits_{i=1}^n \X_{2i}\bm\varepsilon_i^\top \frac{1}{\sqrt{\bm{e}^\top_j{\bm \Sigma}_{\mathcal{S}}^{-1}\bm{e}_j}}\\
&\quad +\left(\frac{1}{\sqrt{\bm{e}^\top_j{\bf A}^{-1}\bm{e}_j}}-\frac{1}{\sqrt{\bm{e}^\top_j{\bm\Sigma}_{\mathcal{S}}^{-1}\bm{e}_j}}\right)\frac{1}{\sqrt{n}}\bm{e}^\top_j{\bm \Sigma}_{\mathcal{S}}^{-1}\sum\limits_{i=1}^n \X_{2\mathcal{S}i}\bm\varepsilon_i^\top\\
&\quad+\frac{1}{\sqrt{n}}\bm{e}^\top_j({\bf A}^{-1}-{\bm \Sigma}_{\mathcal{S}}^{-1})\sum\limits_{i=1}^n \X_{2\mathcal{S}i}\bm\varepsilon_i^\top\left(\frac{1}{\sqrt{\bm{e}^\top_j{\bf A}^{-1}\bm{e}_j}}-\frac{1}{\sqrt{\bm{e}^\top_j{\bm\Sigma}_{\mathcal{S}}^{-1}\bm{e}_j}}\right)
\\
&:=\Lambda_1+\Lambda_2+\Lambda_3,
\end{align*}
where ${\bf A}=n^{-1}\sum\nolimits_{i=1}^n\X_{2\mathcal{S}i}\X_{2\mathcal{S}i}^{\top}$.  Firstly note that
\begin{align*}
\|{\bf A}^{-1}-{\bm\Sigma}_{\mathcal{S}}^{-1}\|_{\infty}\leq \|\bm\Sigma_{\mathcal{S}}^{-1}\|_{\infty}\|{\bf A}^{-1}\|_{\infty}
\|\bm\Sigma_{\mathcal{S}}-{\bf A}\|_{\infty}=O_p(\upsilon_n\bar{q}_n a_{np}).
\end{align*}
Further observe that
\begin{align*}
\left|\frac{1}{\sqrt{n}}\bm{e}^\top_j({\bf A}^{-1}-{\bm\Sigma}_{\mathcal{S}}^{-1})\sum\limits_{i=1}^n \X_{2\mathcal{S}i}\bm\varepsilon_i^\top\right|
\leq \|{\bf A}^{-1}-{\bm\Sigma}_{\mathcal{S}}^{-1}\|_{\infty}\left\|\frac{1}{\sqrt{n}}\sum\limits_{i=1}^n \X_{2\mathcal{S}i}\bm\varepsilon_i^\top\right\|_{\infty}
=O_p(\upsilon_n\bar{q}_n a_{np}\sqrt{\log{\bar{q}_n}}).
\end{align*}
While by Assumption \ref{DesignMatrix}, there exists a constant $c$ such that
\begin{align*}
\left|\frac{1}{\sqrt{\bm{e}^\top_j{\bf A}^{-1}\bm{e}_j}}-\frac{1}{\sqrt{\bm{e}^\top_j{\bm\Sigma}_{\mathcal{S}}^{-1}\bm{e}_j}}\right|
=&\frac{\left|{\bm{e}^\top_j{\bm\Sigma}_{\mathcal{S}}^{-1}\bm{e}_j}
-{\bm{e}^\top_j{\bf A}^{-1}\bm{e}_j}\right|}{\left|\sqrt{\bm{e}^\top_j{\bf A}^{-1}\bm{e}_j}
+\sqrt{\bm{e}^\top_j{\bm\Sigma}_{\mathcal{S}}^{-1}\bm{e}_j}\right|}\left|\frac{1}{\sqrt{\bm{e}^\top_j{\bf A}^{-1}\bm{e}_j}\sqrt{\bm{e}^\top_j{\bm\Sigma}_{\mathcal{S}}^{-1}\bm{e}_j}}\right|\\
\leq & c\|\bm\Sigma_{\mathcal{S}}^{-1}-{\bf A}^{-1}\|_{\infty}=O_p(\upsilon_n\bar{q}_na_{np}).
\end{align*}
Thus, $\Lambda_1=O_p(\upsilon_n\bar{q}_n a_{np}\sqrt{\log \bar{q}_n})$, $\Lambda_2=O_p(\upsilon_n\bar{q}_n a_{np}\sqrt{\log{\bar{q}_n}})$, and $\Lambda_3=O_p(\upsilon_n^2\bar{q}_n^2 a_{np}^2\sqrt{\log{\bar{q}_n}})$ is a small order than $\Lambda_1$ and $\Lambda_2$. By triangle inequality and Lemma \ref{lem2}, $W_j-\widetilde{W}_j=O_p(\sqrt{n}c_{np})O_p(\upsilon_n\bar{q}_n a_{np}\sqrt{\log{\bar{q}_n}})=O_p(c_{np}a_{np}\upsilon_n\bar{q}_n\sqrt{n\log {\bar{q}_n}})$. \hfill$\Box$

\bigskip
\noindent
\textbf{Proof of Lemma \ref{klem3}}.
\noindent

	By Lemma \ref{wtw}, with probability tending to one
\begin{align*}
	\left|{\sum_{j\in \mathcal{A}^c}\mathbb{I}(W_{j}\geq t)}-{\sum_{j\in \mathcal{A}^c}\mathbb{I}(\widetilde{W}_{j}\geq t)}\right|
	&\leq \left|\sum_{j\in \mathcal{A}^c}\left\{\mathbb{I}(\widetilde{W}_{j}\geq t+l_n)-\mathbb{I}(\widetilde{W}_{j}\geq t)\right\}\right|\\
	&\quad+\left|\sum_{j\in \mathcal{A}^c}\left\{\mathbb{I}(\widetilde{W}_{j}\geq t-l_n)-\mathbb{I}(\widetilde{W}_{j}\geq t)\right\}\right|\\
	&:=\Lambda_1+\Lambda_2,
\end{align*}
where $l_n/(c_{np}a_{np}\upsilon_n\sqrt{n\bar{q}_n\log \bar{q}_n})\rightarrow \infty$ as $(n,p)\to\infty$.

Then it follows that
\begin{align*}
	&\mathbb{E}(\Lambda_1)=\mathbb{E}\left\{\sum_{j\in \mathcal{A}^c}\mathbb{I}(t\leq\widetilde{W}_{j}\leq t+l_n)\right\}\\
	&=\sum_{j\in \mathcal{A}^c}\text{Pr}(t\leq\widetilde{W}_{j}\leq t+l_n, \|T_{2j}\|_2\leq b_n)
	+\sum_{j\in \mathcal{A}^c}\text{Pr}(t\leq\widetilde{W}_{j}\leq t+l_n, \|T_{2j}\|_2> b_n)\\
	&\leq \sum_{j\in \mathcal{A}^c}\text{Pr}(t\leq\widetilde{W}_{j}\leq t+l_n, \|T_{2j}\|_2\leq b_n)+o(1),
\end{align*}
where we use Lemma \ref{lem1} to get $\text{Pr}(t\leq\widetilde{W}_{j}\leq t+l_n, \|T_{2j}\|_2> b_n)=o(\bar{q}_n^{-1})$.

{Recall that $U_j\sim N_H(0,\widetilde{\bm\Sigma}_j)$.} By Lemma~\ref{mdm}, it suffices to show that
\begin{align*}
	&\sum\limits_{j\in\mathcal{A}^c}{\rm Pr}\left(t\leq T_{1j}^{\top}U_j\leq t+l_n\right)\\
	&=\sum\limits_{j\in\mathcal{A}^c}{\rm Pr}\left(\frac{t}{\sqrt{T_{1j}^{\top}\widetilde{\bm\Sigma}_jT_{1j}}}\leq \frac{T_{1j}^{\top}U_j}{\sqrt{T_{1j}^{\top}\widetilde{\bm\Sigma}_jT_{1j}}}\leq\frac{t+l_n}{\sqrt{T_{1j}^{\top}\widetilde{\bm\Sigma}_jT_{1j}}}\right)\\
	&=\sum\limits_{j\in\mathcal{A}^c}\mathbb{E}\left\{\Phi\left(\frac{t+l_n}{\sqrt{T_{1j}^{\top}\widetilde{\bm\Sigma}_jT_{1j}}}\right)-\Phi\left(\frac{t}{\sqrt{T_{1j}^{\top}\widetilde{\bm\Sigma}_jT_{1j}}}\right)\right\}\\
	&\leq \sum\limits_{j\in\mathcal{A}^c}\mathbb{E}\left\{\frac{l_n}{\sqrt{T_{1j}^{\top}\widetilde{\bm\Sigma}_jT_{1j}}}\phi\left(\frac{t}{\sqrt{T_{1j}^{\top}\widetilde{\bm\Sigma}_jT_{1j}}}\right)\right\}\\
	&\leq l_n\sum\limits_{j\in\mathcal{A}^c}\mathbb{E}
	\left\{\left[\frac{1}{\sqrt{T_{1j}^{\top}\widetilde{\bm\Sigma}_jT_{1j}}}
	\left(\frac{t}{\sqrt{T_{1j}^{\top}\widetilde{\bm\Sigma}_jT_{1j}}}+\frac{\sqrt{T_{1j}^{\top}\widetilde{\bm\Sigma}_jT_{1j}}}{t}\right)\right]
	\widetilde{\Phi}\left(\frac{t}{\sqrt{T_{1j}^{\top}\widetilde{\bm\Sigma}_jT_{1j}}}\right)\right\}\\
	&\lesssim l_n M^{-1}\log {\bar{q}_n}\sum\limits_{j\in\mathcal{A}^c}\mathbb{E}\left\{\widetilde{\Phi}\left(\frac{t}{\sqrt{T_{1j}^{\top}\widetilde{\bm\Sigma}_jT_{1j}}}\right)
	\right\},
\end{align*}
where $\widetilde{\Phi}(x)=1-\Phi(x)$, $\phi(x)$ and $\Phi(x)$ are the density function and cumulative distribution function of standard normal distribution, respectively. The second to last inequality is due to
\begin{align*}
	\phi(x)<\frac{x^2+1}{x}\widetilde{\Phi}(x), \ \ \mbox{for all} \  x>0.
\end{align*}

Similarly we have
\begin{align*}
	\sum\limits_{j\in\mathcal{A}^c}{\rm Pr}(\widetilde{W}_j >t)=\sum\limits_{j\in\mathcal{A}^c}\mathbb{E}\left\{\widetilde{\Phi}\left(\frac{t}{\sqrt{T_{1j}^{\top}\widetilde{\bm\Sigma}_jT_{1j}}}\right)\right\}\left\{1+o(1)\right\}.
\end{align*}

Note that $h_n$ can be made arbitrarily small 
as $n\to\infty$ in Lemma~\ref{klem2}. By the similar proof in Lemma~\ref{klem2}, the result holds if $c_{np}a_{np}\upsilon_n\sqrt{n\bar{q}_n\log \bar{q}_n}\log{\bar{q}_n}h_n\to 0$. The part of $\Lambda_2$ is proved similarly as the part of $\Lambda_1$. Accordingly, we can claim the assertion.
\hfill$\Box$

\subsection*{S4. Additional simulations}

\begin{figure}[htbp]
	\centering
	\includegraphics[scale=0.48]{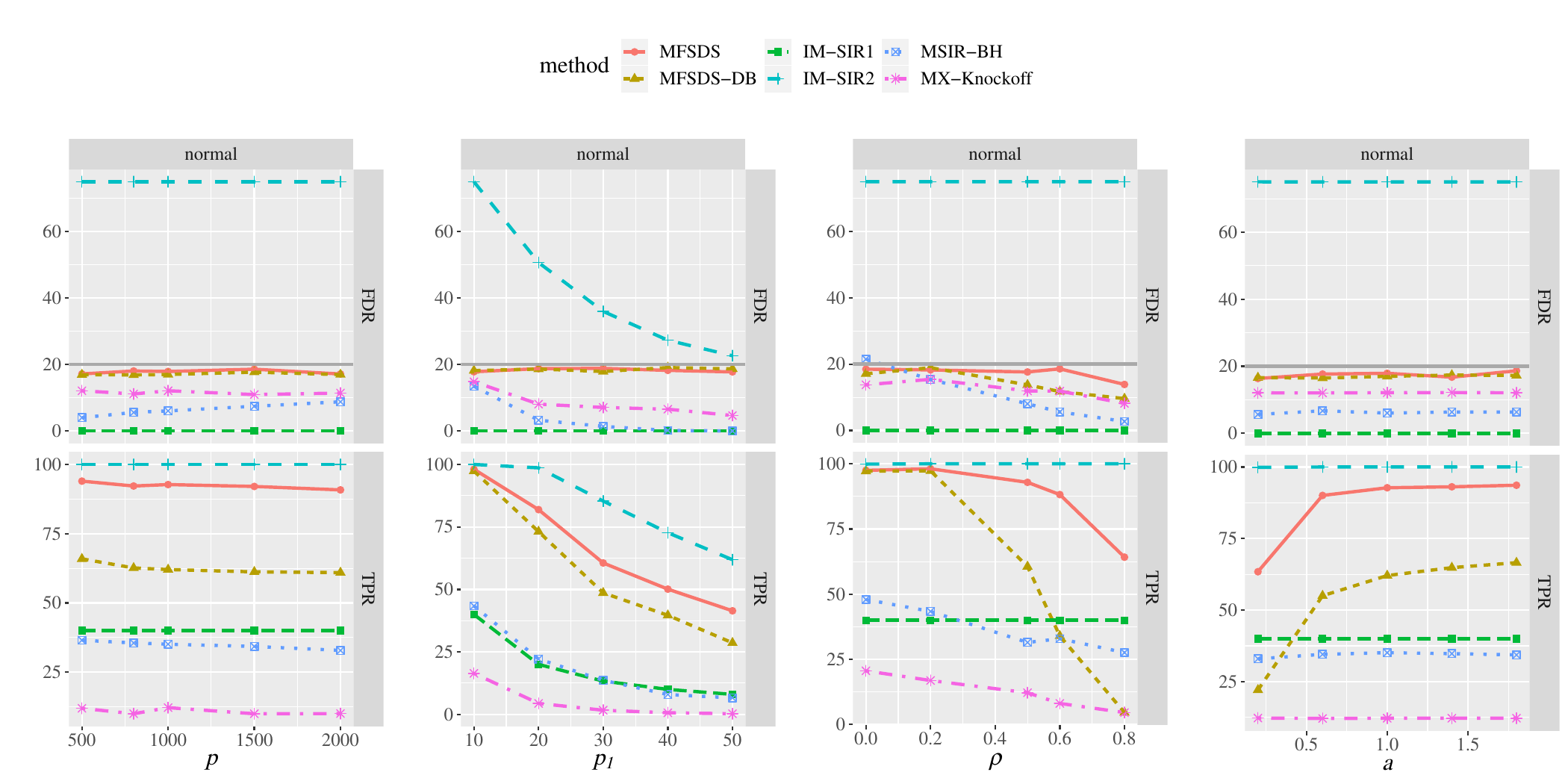}
	\caption{
		FDR and TPR (\%) curves against different covariate dimension $p$, signal number $p_1$,
		correlation $\rho$,  and signal strength $a$ under Scenario 2c when $n=500$ and $\X$ is from normal distribution. 
		The gray solid line denotes the target FDR level.
	}
	\label{highdim-prhom-norm}
\end{figure}

\begin{table}[htbp]
	\caption{FDR, TPR,  $P_a=\text{Pr}(\mathcal{A}\subseteq\widehat{\mathcal{A}}(L))$(\%) and computing time (in second) for several methods against different $\X$ distributions under Scenarios 2a$-$2c when $(n,p,p_1,\rho, a)=(800,1000,10,0.5,1)$.}
	\vspace{0.2cm}
	\label{highdimension-Xd-n800}
	\centering
	{\centering
		\scalebox{1}{
			\begin{tabular}{clrrrrrrrrrrrrrrrr}
				\toprule[1pt]
				&	&  \multicolumn{4}{c}{\textbf{normal}} &&\multicolumn{4}{c}{\textbf{mixed}}\\
				\cline{3-6} \cline{8-11}
				\textbf{Scenario} &\textbf{Method}& FDR&TPR&$P_a$&time&&FDR&TPR&$P_a$&time\\
				
				\hline
				&MFSDS&18.5&100.0&100.0&18.6&&17.9&100.0&99.6&24.2 \\
				&MFSDS-DB&18.1&93.7&63.4&213.5&&18.1&96.7&79.0&209.9\\
				&MSIR-BH&2.7&37.0&0.0&14.1&&4.4&38.7&0.0&20.2\\
				2a & IM-SIR1&0.0&50.0&0.0&28.5&&0.0&50.0&0.0&34.8\\
				&  IM-SIR2&83.1&100.0&100.0&28.5&&83.1&100.0&100.0&34.6\\
				& MX-Knockoff&71.8&6.0&0.0&51.6&&27.4&11.9&0.0&56.7\\
				\hline
				&MFSDS&17.7&99.5&94.6&18.4&&17.9&99.5&95.2&27.7\\
				&MFSDS-DB&17.3&91.6&41.6&215.7 &&17.9&93.6&56.2&234.8\\
				&MSIR-BH&4.4&39.5&0.0&14.0&&4.6&39.1&0.0&20.8\\
				2b & IM-SIR1&0.0&50.0&0.0&28.6&&0.0&50.0&0.0&35.5 \\
				&  IM-SIR2&83.1&100.0&100.0&28.6&&83.1&100.0&100.0&35.9 \\
				& MX-Knockoff&10.0&13.9&0.0&51.6&&31.3&23.6&0.0&57.7  \\
				
				\hline
				&MFSDS&17.6&99.1&92.8&21.0&&16.8&98.9&90.6&38.1\\
				&MFSDS-DB&18.0&85.6&25.2&203.4&&16.6&89.8&38.6&248.7\\
				&MSIR-BH&5.2&38.4&0.0&16.1&&4.7&37.8&0.0&30.0\\
				2c & IM-SIR1&0.0&50.0&0.0&24.4&&0.0&50.0&0.0&40.8\\
				&  IM-SIR2&83.1&100.0&100.0&24.2&&83.1&100.0&100.0&41.3\\
				& MX-Knockoff&10.3&11.1&0.0&54.9&&31.5&23.6&0.0&60.5\\
				
				\bottomrule[1pt]
	\end{tabular}}}\\
\end{table}

\begin{table}[htbp]
	\tabcolsep 6pt
	\caption{FDR, TPR (\%) and computing time (in second) for \cite{guo2024model}  against different signal strengths under Scenarios 2c when $(p,p_1,\rho) = (1000,10,0.5)$.}
	\vspace{0.2cm}
	\label{Tab:Guo}
	\centering
	{\centering
		\scalebox{1}{
			\begin{tabular}{clrrrrrrrrrrrrrrrr}
				\toprule[1.5pt]
				& &&  \multicolumn{3}{c}{\textbf{normal}} &&\multicolumn{3}{c}{\textbf{mixed}}\\
				\cline{4-6} \cline{8-10}
				$a$ &	\textbf{Method}&&FDR&TPR&time&&FDR&TPR&time\\
				\hline
				&MFSDS&&16.4&63.4&12.9&&16.6&63.0&20.3\\
				$0.2$&MFSDS-DB&&16.6&22.1&100.6&&15.6&30.3&84.7\\
				&MSIR-BH&&5.6&33.0&13.1&&5.9&32.4&17.2\\
				& \cite{guo2024model} && 7.1&47.6&1559.4&&5.0&48.0&2494.9 \\
				\hline
				&MFSDS&&17.9&92.7&12.4&&19.2&92.8&22.8\\
				$1$&MFSDS-DB&&17.0&62.1&102.3&&17.4&73.2&126.8\\
				&MSIR-BH&&6.5&33.4&12.2&&6.3&31.1&21.8\\
				&\cite{guo2024model} &&   6.5&94.3&1668.6&&7.5&94.3&2712.1\\
				\bottomrule[1.5pt]
	\end{tabular}}}\\
\end{table}

\end{document}